\documentclass{article}
\usepackage{graphicx} 
\def\withcolors{0}
\def\withnotes{0}
\usepackage{ccanonne}
\usepackage{multicol}
\usepackage{multirow}
\usepackage{subcaption}

\newtheorem{fact}[theorem]{Fact}
\allowdisplaybreaks
\newcolumntype{C}{>{\centering\arraybackslash}p{0.27\textwidth}}
\usepackage{thmtools}
\usepackage{thm-restate}

\newcommand{\dims}{d}

\newcommand{\opnorm}[1]{{\left\|#1\right\|}_{\text{op}}}
\newcommand{\tracenorm}[1]{{\left\|#1\right\|}_{1}}
\newcommand{\hsnorm}[1]{{\left\|#1\right\|}_{\text{HS}}}
\newcommand{\barDelta}{{\bar{\Delta}}}
\newcommand{\ptb}{{z}}
\newcommand{\ptbDistr}{{\mathcal{D}_\ell}}

\newcommand{\supparen}[1]{^{(#1)}}

\newcommand{\cd}{{c}}
\newcommand{\cop}{\kappa}

\newcommand{\ismixed}{\texttt{YES}}
\newcommand{\notmixed}{\texttt{NO}}

\newcommand{\bfP}{\mathbf{P}}
\newcommand{\bfQ}{\mathbf{Q}}

\newcommand{\x}{\mathbf{x}}

\def\multiset#1#2{\ensuremath{\left(\kern-.3em\left(\genfrac{}{}{0pt}{}{#1}{#2}\right)\kern-.3em\right)}}

\newcommand{\qs}{\rho}
\newcommand{\qmm}{{\rho_{\text{mm}}}}
\newcommand{\qkn}{{\rho_0}}

\newcommand{\HH}{\mathbb{H}}
\newcommand{\Herm}[1]{{\HH_{#1}}}

\newcommand{\qbit}[1]{|#1\rangle}
\newcommand{\qadjoint}[1]{\langle#1|}
\newcommand{\qproj}[1]{\qbit{#1}\qadjoint{#1}}
\newcommand{\qoutprod}[2]{\qbit{#1}\qadjoint{#2}}

\newcommand{\qdotprod}[2]{\langle#1|#2\rangle}
\newcommand{\hdotprod}[2]{\left\langle#1,#2\right\rangle}
\newcommand{\matdotprod}[3]{\langle#1|#2|#3\rangle}

\newcommand{\eye}{\mathbb{I}}
\newcommand{\VecOp}{\text{vec}}
\newcommand{\Real}{\text{Re}}
\newcommand{\Img}{\text{Im}}

\newcommand{\bx}{\mathbf{x}}

\newcommand{\Luders}{\mathcal{H}}
\newcommand{\avgLuders}{{\bar{\Luders}}}
\newcommand{\Choi}{{\mathcal{C}}}
\newcommand{\avgChoi}{{\bar{\Choi}}}
\newcommand{\hbasis}{{\mathcal{V}}}

\DeclareMathOperator{\diag}{diag}

\DeclareMathOperator{\Span}{span}

\newcommand{\POVM}{\mathcal{M}}

\title{The role of shared randomness in quantum state certification with unentangled measurements}
\author{\begin{tabular}{C C }
  Yuhan Liu$^{\textcolor{red}{*}}$  & Jayadev Acharya\thanks{This research is supported by NSF-CCF-1815893, NSF-CCF-1846300 (CAREER), and a Google Research Scholar Grant.}  \\
 Cornell University & Cornell University \\ 
\small \texttt{yl2976@cornell.edu} &\small \texttt{acharya@cornell.edu} 
\end{tabular}
}

\begin{document}

\maketitle
\begin{abstract}
Given $n$ copies of an unknown quantum state $\rho\in\mathbb{C}^{d\times d}$, quantum state certification is the task of determining whether $\rho=\rho_0$ or $\|\rho-\rho_0\|_1>\varepsilon$, where $\rho_0$ is a known reference state.
We study quantum state certification using unentangled quantum measurements, namely measurements which operate only on one copy of $\rho$ at a time.
When there is a common source of shared randomness available and the unentangled measurements are chosen based on this randomness, prior work has shown that $\Theta(d^{3/2}/\varepsilon^2)$ copies are necessary and sufficient. This holds even when the measurements are allowed to be chosen adaptively.
We consider deterministic measurement schemes (as opposed to randomized) and demonstrate that ${\Theta}(d^2/\varepsilon^2)$ copies are necessary and sufficient for state certification. This shows a separation between algorithms with and without shared randomness. 

We develop a unified lower bound framework for both fixed and randomized measurements, under the same theoretical framework that relates the hardness of testing to the well-established L\"uders rule. More precisely, we obtain lower bounds for randomized and fixed schemes as a function of the eigenvalues of the L\"uders channel which characterizes one possible post-measurement state transformation.
    
\end{abstract}

\section{Introduction}
We study the problem of quantum state certification~\cite{ODonnellW15,wright2016learn, BadescuO019}, 
where we are given $\ns$ copies of an \textit{unknown} quantum state with density $\qs\in\C^{\dims\times \dims}$, {and complete description of a known state $\qkn$}. 
The goal is to use quantum measurements to test whether $\qs = \qkn$ or $\tracenorm{\qs-\qkn}>\eps$, where $\tracenorm{\cdot}$ is the trace norm.  
A special case of this problem is \emph{mixedness testing}, which is the case when $\qkn=\qmm \eqdef\eye_{\dims}/\dims$ is the maximally mixed state. 
Quantum certification is motivated by practical applications where one wants to verify whether the output state of a quantum algorithm is indeed the state we desire.

{A related problem is of \emph{closeness testing}, where we are given copies of two \emph{unknown} states $\qs$ and $\qkn$ and the goal is again to test whether $\qs = \qkn$ or $\tracenorm{\qs-\qkn}>\eps$. The motivation to study this problem is to test whether two quantum algorithms produce the same state.}

We are interested in determining how many copies of the unknown state(s) are needed to perform the task of testing. This task of understanding the copy complexity quantum state certification  was studied initiated in~\cite{ODonnellW15} and later in~\cite{BadescuO019}. They showed that when we are allowed to perform arbitrary entangled quantum measurements over the $\ns$ copies, then $\ns=\Theta(d/\eps^2)$ copies are necessary and sufficient for testing. However, entangled measurements are currently infeasible to implement in practice, even for moderate values of $\ns$ and $\dims$.  It is desirable to use unentangled measurements, where a quantum measurement is done on one copy of $\rho$ (and $\qkn$ if it is also unknown) at a time. Such \emph{unentangled measurements} (also referred to as \emph{incoherent} and \emph{independent} in previous literature) can be categorized into three different protocols:

\begin{enumerate}
\item \textbf{Fixed/Deterministic measurements.} The set of measurements(POVMs) to be performed are fixed \emph{ahead of time}. Once the copies of the quantum states are available, we use these fixed measurements for the task of testing. A key advantage of such protocols is that the same set of measurements can be used for multiple repetitions of the testing problem. Moreover, there is no latency since the measurements are not designed after the states are made available, which is a drawback of the following protocols.

\item \textbf{Randomized non-adaptive measurements.} In this setting, there is common randomness available, and the set of measurements at the different copies are all chosen simultaneously as a function of this common randomness.
Every time we want to test for a state, we need to instantiate the common randomness and select the set of measurements using a new instance of the common randomness. This is done after the copies of the state are made available.\footnote{If the set of measurements is finite, we can prepare all measurements beforehand and sample with classical randomness. However, this could still be difficult if the set is very large.}
\item \textbf{Randomized adaptive measurements.} In this setting common randomness is available across the measurements. Furthermore, the measurements are applied sequentially to each copy of $\qs$, and the measurement on the next copy of $\qs$ can depend on the outcome of previous measurements\footnote{In principle, one can use the first measurement outcome as a source of shared randomness for all other measurements, so adaptive measurements are essentially randomized.}. 
A primary drawback of this scheme is the latency and complications associated with designing measurements one after another.
\end{enumerate}

\subsection{Prior Works}
{\cite{ODonnellW15} initiated the study of copy complexity of the task of quantum state certification. They considered entangled measurements and showed that  $\ns=\Theta(d/\eps^2)$ copies are necessary and sufficient for testing. This is also the copy complexity of closeness testing~\cite{BadescuO019}.}

{Given the practical relevance of unentangled measurements, it has been considered in several prior works. For the task of quantum mixedness testing,~\cite{BubeckC020} showed that when randomized non-adaptive unentangled measurements are allowed, then $\ns=\Theta(d^{3/2}/\eps^2)$ copies are necessary and sufficient. \cite{ChenLO22instance} extended the results to the cases when $\qkn$ need not be the maximally mixed state, and also when it is unknown (closeness testing).}

The role of adaptivity was studied in~\cite{Chen0HL22}, who showed that even with adaptive measurements, the number of copies necessary is still $\ns=\Omega(d^{3/2}/\eps^2)$. In other words, adaptivity does not help improve the copy complexity of testing.
~\cite{Yu2023almost} achieved $\ns =\tilde{\Theta}(\dims^2/\eps^2)$ using randomly sampled Pauli measurements, which are more restrictive yet easier to implement. 
The drawback of all the algorithms proposed in these works is the necessity of randomization in the measurements.

\paragraph{Quantum tomography.}
In \emph{quantum tomography},  the goal is to estimate the unknown state $\qs$ to within $\eps$ in trace distance.~\cite{ODonnellW16,HaahHJWY17} established the optimal copy complexity for this task as $\Theta(\dims^2/\eps^2)$ with entangled measurements. For unentangled measurements, various works~\cite{kueng2017low,HaahHJWY17,chen2023does} have shown that for quantum tomography, $\Theta(\dims^3/\eps^2)$ are necessary and sufficient to estimate a full-rank $\rho$, even when adaptivity is allowed~\cite{chen2023does}.~\cite{guctua2020fast} showed that the bound is achievable up to log factors using \emph{fixed} structured POVMs (e.g. SIC-POVM~\cite{zauner1999grundzuge,renes2004symmetric}, maximal MUB~\cite{klappenecker2005mutually}).

\subsection{New results}

We consider state certification with fixed measurements, where the POVMs are fixed ahead of time, and can be used for multiple repetitions of the problem.

The naive solution is to apply the fixed unentangled measurements for quantum tomography~\cite{guctua2020fast} giving an upper bound of $\ns=\bigO{d^{3}/\eps^2}$. However, since tomography is strictly harder than testing, we might expect to do much better than $\dims^3$.
Indeed,~\cite{Yu21sample} designed a simple algorithm with fixed measurements that achieves $O(\dims^2/\eps^2)$ copy complexity. 
The lower bound, however, was left as an outstanding open problem. Without shared randomness, it is unknown whether $\ns=O(d^{3/2}/\eps^2)$ copies are sufficient to perform quantum state certification, or if we need more copies due to the lack of randomness. 

{We establish the copy complexity of quantum state certification with fixed unentangled measurements. Our main result, stated below, shows that indeed there is a cost in copy complexity that we have to pay for having schemes that are \emph{fixed and reusable}. Please see Section~\ref{sec:set-up} for the formal problem definition.

\begin{theorem}
\label{thm:main}
For unentangled POVMs without randomness, $\ns =\Theta(\dims^2/\eps^2)$ copies are necessary and sufficient to test whether $\qs=\qkn$ or $\tracenorm{\qs-\qkn}>\eps$ in the trace distance with probability at least 2/3. 
\end{theorem}

\cref{tab:results} places our work in the context of existing results for other types of measurements. There is a strict $\Theta(\sqrt{\dims})$-factor separation between fixed and randomized non-adaptive schemes. We note that the shared randomness source can be entirely independent of the quantum states, so it is in some sense surprising that a piece of irrelevant random information leads to substantial improvement in copy complexity. This demonstrates that shared randomness is a valuable and important resource in unentangled quantum state certification.
\begin{table}[H]
    \def\arraystretch{1.6}
        \centering
        \begin{tabular}{|c | c | c | c |}
        \hline
       \multicolumn{2}{|c|}{Measurement type}  &  Upper bound & Lower bound \\ \hline
            \multicolumn{2}{|c|}{Entangled} & $\frac{\dims}{\eps^2}$ \cite{ODonnellW15, BadescuO019}  & $\frac{\dims}{\eps^2}$ \cite{ODonnellW15} \\\hline
            \multirow{3}{*}{Unentangled} & Adaptive  &  $\frac{\dims^{3/2}}{\eps^2}$  \cite{BubeckC020, ChenLO22instance} & $\frac{\dims^{3/2}}{\eps^2}$\cite{Chen0HL22}\\ \cline{2-4}
            & Randomized &  $\frac{\dims^{3/2}}{\eps^2}$ \cite{BubeckC020,ChenLO22instance}  &  $\frac{\dims^{3/2}}{\eps^2}$ \cite{BubeckC020} \\ \cline{2-4}
            & Fixed &  $\frac{\dims^2}{\eps^2}$ \cite{Yu21sample} & $\frac{\dims^2}{\eps^2}$ \textbf{(This work)}\\\hline
        \end{tabular}
        \caption{Existing and new worst-case copy complexity results for quantum state certification.}
        \label{tab:results}
    \end{table}
    
We develop an information-theoretic framework that provides a unified lower bound for testing with unentangled non-adaptive measurements.  Using this framework, we obtain a simpler proof (in the sense that we do not need Weingarten calculus) for the lower bound of $\ns=\Omega(d^{3/2}/\eps^2)$ established in~\cite{BubeckC020} for randomized measurements.

\subsection{Our techniques}

Our main contribution is a unified lower bound framework for quantum state certification that works for both randomized and fixed non-adaptive schemes. Before we introduce the technical contributions, we provide a high-level explanation of why shared randomness is crucial in unentangled quantum testing and motivate with a simple example.

\paragraph{The disadvantage of fixed measurements through a simple example.}
\label{sec:fixed-disadvantage}

In randomized schemes, given the copies of the state, we then choose the measurements randomly. However, for fixed measurements, the measurement scheme is fixed and the state is then chosen. Thus, without randomness, nature would have the opportunity to \textit{adversarially} design a quantum state that fools the pre-defined set of measurements. When shared randomness is available, we can avoid the bad effect of adversarial choice of quantum states. In principle, this qualitative gap is like the difference between randomized algorithms and deterministic algorithms.

We use a simple example to demonstrate this idea. Suppose we choose each measurement $\POVM_i$ simply to be the same canonical basis measurement, i.e. $\POVM_i=\{\qproj{x}\}_{x=0}^{\dims-1}$. Then nature can set $\rho$ to be the ``+'' state where
\begin{equation}
    \rho = \qproj{\phi},\quad \qbit{\phi} = \frac{1}{\sqrt{\dims}}\sum_{x=0}^{\dims-1}\qbit{x}.
    \label{equ:plus-state}
\end{equation}

Note that the trace distance $\frac{1}{2}\tracenorm{\rho-\qmm}=1-1/\dims\simeq 1$. When the underlying state is $\rho$, all measurement outcomes $x_i$ would be independent samples from the uniform distribution over $\{0, \ldots, d-1\}$. However, if the state is the maximally mixed state $\qmm$, the distribution of each measurement outcome would also be the uniform distribution over $\{0, \ldots, d-1\}$. Thus, even though the trace distance between $\rho$ and $\qmm$ is large, the measurement scheme is completely fooled.

On the other hand, with shared randomness, one can (theoretically) sample a basis uniformly from the Haar measure as in \cite{BubeckC020} to easily avoid this issue. No fixed $\rho$ would be able to completely fool the randomized basis measurement sampled uniformly. In fact with high probability, the randomly sampled basis is good in the sense that the outcome distribution would be far enough when the two states $\rho$ and $\qmm$ are far (see~\cite[Lemma 6.3]{ChenLO22instance}).

\paragraph{A novel lower bound construction.} 

The design of hard instances has to account for the difference illustrated above when proving lower bounds for randomized and fixed measurements. In particular, for randomized measurements, the lower bound construction can be \emph{measurement independent}. However, for fixed measurements, since the states can be chosen adversarially, the lower bound construction needs to be \emph{measurement-dependent}.

Many prior works on testing and tomography~\cite{ODonnellW15,ODonnellW16,HaahHJWY17, BubeckC020, ChenCH021,Chen0HL22} use \textit{measurement-independent} distributions over states in $\mathbb{C}^{\dims\times\dims}$ to prove lower bounds. In particular, \cite{BubeckC020, Chen0HL22} show that testing within a specific class requires at least $\ns=\Omega(d^{3/2}/\eps^2)$ when working with \emph{randomized} and \emph{adaptive} unentangled measurements respectively. Unfortunately, for these measurement-independent constructions, there must exist fixed measurement schemes whose copy complexity is $\ns=O(d^{3/2}/\eps^2)$ due to standard derandomization arguments. To prove a stronger lower bound for fixed measurements, our task is different. We must show that for any fixed measurement scheme, we can design a hard instance of the testing problem that would require at least $\ns=\Omega(\dims^2/\eps^2)$. We note that the lower bound construction in~\cite{Yu2023almost} is measurement-dependent, but specifically tailored to Pauli measurements and not general enough for our purpose.

Our generic \textit{measurement-dependent} lower bound construction is a necessary and novel contribution that leads to tight lower bounds for state certification with fixed measurements. It takes the form
\[
\sigma_z=\qmm + \frac{\eps}{\sqrt{\dims}}\cdot\frac{\cd}{\dims}\sum_{i=1}^{\dims^2/2} z_iV_i,
\]
where $\{V_i\}_{i=1}^{\dims^2-1}$ are $\dims^2-1$ orthonormal trace-0 Hermitian matrices, $z=(z_1, \ldots, z_{\dims^2/2})$ are uniformly sampled from $\{-1, 1\}^{\dims^2/2}$, and $\cd$ is an absolute constant. 
In essence, we perform independent binary perturbations along different trace-0 directions. We show with appropriate choice of $c$, regardless of the choice of $\{V_i\}_{i=1}^{\dims^2-1}$, 
with high probability over the randomness of $z$, $\sigma_z$ is a valid quantum state and $\eps$-far in trace distance from $\qmm$.
The matrices $V_i$'s can be chosen \emph{dependent} on the fixed measurement scheme that we want to \emph{fool}.
In particular, the perturbations $V_1, \ldots, V_{\dims^2/2}$ can be chosen in directions about which the fixed measurement schemes provide the least information.
The matrices $V_i$'s can also be fixed, in which case the construction is measurement-independent and our framework naturally leads to the lower bound for randomized non-adaptive measurements in~\cite{BubeckC020}.

The binary perturbations in our construction are mathematically easier to handle. In prior works~\cite{BubeckC020, HaahHJWY17, ChenCH021} designed the hard cases using random unitary transformations around the maximally-mixed state, which requires difficult calculations using Weingarten calculus~\cite{weingarten1978asymptotic,collins2003moments}. In contrast, our arguments avoid the difficult representation-theoretic tools. \cite{Chen0HL22, chen2023does} used Gaussian orthogonal ensembles, which perturbs each matrix entry with independent Gaussian distributions. Binary perturbations share many statistical similarities with Gaussian since both are sub-gaussian distributions. However, the former is arguably simpler as the support is finite, and thus information-theoretic tools can be more easily applied.

We note that these constructions are all in spirit motivated by lower bounds in classical discrete distribution testing where the hard instances are constructed as perturbations around the uniform distribution~\cite{Paninski08}.

\paragraph{Unified lower bound framework.}
For both fixed and randomized schemes, we find that perhaps very coincidentally or very naturally, the ability of a measurement scheme to distinguish between quantum states is characterized by the eigenvalues of the L\"uders channel~\cite{debrota2019luders}\footnote{This was referred to as the expected density operator in~\cite[Section 3.3]{khatri2021principles} } $\Luders:\C^{\dims\times\dims}\mapsto\C^{\dims\times\dims}$ defined by \textit{the L\"uders rule}~\cite{luders1950zustandsanderung},
\[
\Luders(\rho)\eqdef \sum_{x}\sqrt{M_x}\rho\sqrt{M_x}.
\]
where $\{M_x\}_x$ is a rank-1 POVM and $\sqrt{M_x}$ is the unique p.s.d. square root of $M_x$.  This channel describes one possible post-measurement state of the POVM $\{M_x\}_x$ if the measurement outcome is not available. 
Indeed, if two quantum states $\rho$ and $\sigma$ yield the same post-measurement state, it is natural to believe that the two states cannot be distinguished by the measurement scheme. 
This is in spirit similar to the chi-squared contraction framework in~\cite{AcharyaCT19} which was used to derive a unified lower bound for information-constrained inference of classical distributions in the distributed setting.
Depending on whether shared randomness is available, the copy complexity depends on different norms of the $\Luders$ channel, as shown in~\cref{tab:result-luders}, where $\tracenorm{\Luders}$ and $\hsnorm{\Luders}$ are the trace and Hilbert-Schmidt/Frobenius norms} of $\Luders$. The precise statement is stated in Theorem~\ref{thm:maxmin-and-minmax-ub}. 

\begin{table}[H]
    \centering
    \def\arraystretch{2}
    \begin{tabular}{|c|c|c|}
    \hline
         &  Fixed & Randomized\\ \hline
        Lower bound & $\frac{\dims^2}{\eps^2}\cdot \frac{\dims}{\max_{\Luders}\tracenorm{\Luders}}$ & $\frac{\dims^2}{\eps^2\max_{\Luders}\hsnorm{\Luders}}$ \\ \hline
    \end{tabular}
    \caption{Copy complexity lower bound of non-adaptive state certification in terms of the L\"uders channel.}
    \label{tab:result-luders}
\end{table}

It is straightforward to prove that $\hsnorm{\Luders}^2\le\tracenorm{\Luders}\le \dims$, and thus we obtain tight copy complexity lower bounds for both fixed and randomized non-adaptive measurements.

Our lower bound results yield a very natural quantum interpretation. We believe that a similar characterization using L\"uders channel could be applied to adaptive measurements and other problems such as quantum tomography, similar to~\cite{AcharyaCT19,ACLST22iiuic} which developed a unified framework for distributed learning and testing of discrete distributions.

\subsection{Related works}
\paragraph{Learning information about quantum states.} The literature is vast and here we only discuss the most relevant ones. Our work falls into the line of quantum state certification~\cite{ODonnellW15,BubeckC020,Chen0HL22}. In addition to worst-case bounds that depend on $\dims$, \cite{ChenLO22instance,Chen0HL22} considered general quantum state certification where the copy complexity decreases when the reference state is approximately low rank. Other closeness measures such as fidelity and Bures $\chi^2$-divergence were also considered in~\cite{BadescuO019}. Many works have studied other related problems such as closeness testing~\cite{BadescuO019,Yu21sample, Yu2023almost} (test whether \textit{two unknown} states $\rho$ and $\sigma$ are equal or $\eps$ far), hypothesis testing~\cite{ogawa2000strong, brandao2020adversarial,regula2023postselected} (distinguish between two known states), and hypothesis selection~\cite{BadescuO21, fawzi:hal-04107265} (determine $\rho$ from a finite set of hypothesis sets).

Quantum tomography \cite{ODonnellW16, wright2016learn, ODonnellW17,HaahHJWY17,guctua2020fast,flamian2023tomography} aims to learn a full description of quantum states and can naturally be applied to state certification and all other problems discussed above, although likely with far more copies than needed. Shadow tomography~\cite{Aaronson20,huang2020predicting,BrandaoKLLSW19} considers a much simpler problem of learning the statistics of the state $\rho$ over a finite set of observables. Algorithms for shadow tomography can be applied to quantum hypothesis selection~\cite{BadescuO21, fawzi:hal-04107265}.

In addition to the four types of measurements discussed before, \textit{Pauli} measurements have also attracted significant interest~\cite{flamia2011direct, Liu2011universal,  cai2016optimal,Yu2023almost} due to ease in implementation despite being less powerful. Moreover, \cite{fawzi:hal-04107265} considered \textit{sequential} strategies which allow the number of measurements to depend on previous outcomes (e.g. one can choose to stop measuring the remaining copies if the outcomes so far yield a good estimate), which is parallel to the adaptivity of measurements.

\paragraph{Classical distribution testing.} Quantum state certification can be viewed as the quantum equivalent of testing identity of discrete distributions from samples. Here the task is to decide from samples whether a distribution is equal to a given known distribution. The problem has been well studied starting with the works of~\cite{BatuFFKRW01, Paninski08} which establish the sample complexity of this task when all the samples are available. This is similar to using entangled measurements in the quantum case. Recently there has been significant work on distributed testing of distributions, where instead of having all samples at the same place, they are distributed across users, and we obtain only limited information about each sample, e.g., a communication-constrained~\cite{barnes2019lower, ACT:19:IT2}, or privacy-preserving information~\cite{duchi2013local, ACFT:19:IT3, han2015minimax}. Thinking of each sample analogous to one copy, this distributed testing is in spirit similar to unentangled measurements, where we perform measurements on one copy at a time.~\cite{AcharyaCT19,ACLST22iiuic} derived a unified information-theoretic framework in terms of the channel constraints. In particular,~\cite{AcharyaCT19, ACT:19:IT2, ACFT:19:IT3} showed that there is a separation for distributed testing under communication and privacy constraints between the cases when common randomness was available versus not. Furthermore,~\cite{ACLST22iiuic} show that adaptivity does not help in these problems beyond common randomness. Our results are qualitatively similar to these classical distributed testing results. We show in this work that these ideas can be generalized to quantum state certification and a similar separation also holds. We refer the readers to~\cite{canonne2022topics} for a comprehensive survey of the above topics.

\paragraph{Outline.} The rest of the paper is organized as follows. In Section~\ref{sec:preliminaries}, we give the precise problem definition and provide some mathematical terminology and definitions. In Section~\ref{sec:lower} we introduce our unified lower bound framework for non-adaptive measurements. In Section~\ref{sec:upper-bound} we describe the algorithm that achieves the copy complexity upper bound for fixed measurements.

\section{Preliminaries}
\label{sec:preliminaries}
\subsection{Basics of quantum computing}
\paragraph{Hilbert space over complex vectors.} The space of $\dims$-dimensional complex vectors $\C^\dims$ forms a Hilbert space. We will frequently use the Dirac notation. First, $\qbit{j}\in \R^\dims, j=0, \ldots,\dims-1$ represents a vector with 0 everywhere except the index $j$. These $\dims$ vectors form the canonical basis for $\C^\dims$. 

We use $\qbit{\psi}\in\C^\dims$ to denote a vector. It can be written as a linear combination of the basis, $\qbit{\psi}=\sum_{j=0}^{\dims-1}\psi_j\qbit{j}$ where $\psi_j\in \C$. 
We use $\qadjoint{\psi}$ to denote the adjoint (conjugate transpose) of $\qbit{\psi}$, which is a row vector. The inner product between two vectors $\qbit{\psi}$ and $\qbit{\phi}$ is defined as $\qdotprod{\psi}{\phi}=\sum_{j=0}^{\dims-1}\bar{\psi}_j\phi_j$. 

We can define tensor product where $\qbit{i}\otimes\qbit{j}\in\C^{\dims^2}$ is a vector with 1 at index $i\cdot\dims+j$ and 0 everywhere else.  For general vectors, $\qbit{\psi}\otimes\qbit{\phi}\eqdef\sum_{i, j=0}^{\dims-1}\psi_i\phi_j\qbit{i}\otimes\qbit{j}$. We sometimes omit the $\otimes$ sign for convenience. 
\paragraph{Quantum states.} In a $\dims$-dimensional quantum system, the state $\rho$ is a $\dims\times \dims$ positive semi-definite Hermitian matrix with $\Tr[\rho]=1$. If the rank of $\rho$ is 1, then the state is called \emph{pure} state and $\rho=\qproj{\psi}$ for some unit-norm $\qbit{\psi}\in \C^{\dims}$. Otherwise, the state is a \emph{mixed state}. A special case is $\qmm\eqdef \eye_\dims/\dims$, which is the maximally mixed state.
\paragraph{Measurements.} All measurements can be formulated as a \emph{positive operator-valued measure} (POVM). Let $\mathcal{X}$ be a finite set of outcomes. Then a POVM $\POVM=\{M_x\}_{x\in \mathcal{X}}$, where $M_x$ is p.s.d. and $\sum_{x\in \mathcal{X}}M_x=\eye_\dims$. Let $X$ be the outcome when applying measurement $\POVM$ to $\rho$, then the probability of observing $x$ is given by the \emph{Born's rule},
\[
\probaOf{X=x}=\Tr[\rho M_x].
\]
We note that the outcome set $\mathcal{X}$ need not be finite, in which case POVMs and Born's rule can be generalized. However, finite POVMs are without loss of generality. In principle, all physically feasible measurements are finite. Moreover, our argument extends easily to infinite POVMs.
\subsection{Problem setup}
\label{sec:set-up}

We are given $\ns$ independent copies of an unknown $\dims-$dimensional quantum state $\rho\in\mathbb{C}^{\dims\times\dims}$. The goal is to design 
\begin{itemize}
    \item $\ns$ POVMs  $\mathcal{M}^n=(\POVM_1, \ldots, \POVM_\ns)$ that are applied to the $\ns$ copies of the state that produce the measurement outcomes $\bx = (x_1, \ldots, x_\ns)$, 
    \item a tester $T$ such that when $\qs=\qkn$ it outputs $\ismixed$ with probability at least 2/3 and it outputs $\notmixed$ with probability at least $2/3$ if $\qs$ is at least $\eps$ away from $\qkn$ in trace distance. In other words, we want
\[
\Pr_{\rho = \qkn}(T(\bx) = \ismixed) \ge \frac23,\ \  \text{and} \inf_{\rho:\tracenorm{\rho-\qkn}>\eps}\Pr(T(\bx) = \notmixed) \ge \frac 23.
\]
\end{itemize}
When $\qkn=\qmm\eqdef\eye_\dims/\dims$, the problem is called \emph{mixedness testing}. The smallest value of $\ns$ for which we can design such a tester for all $\qkn$ is the \emph{copy complexity} of quantum state certification.

We apply measurements for each copy individually. More precisely, for the $i$-th copy, we apply a POVM $\POVM_i=\{M_x^i\}_{x=1}^k$ where $M_x^i$ is p.s.d. and $\sum_{x}M_x^i=\eye_\dims$. Let $x_i$ be the measurement outcome after applying $\POVM_i$ on the $i$-th copy. When the quantum state is $\rho$, $x_i$ follows a discrete distribution $\p^i_{\rho}=[\p^i_{\rho}(1), \ldots, \p^i_{\rho}(k)]$ given by \textit{Born's rule},
\begin{equation*}
    \p^i_{\rho}(x) = \Tr[M_x^i\rho], \quad  x=1, \ldots, k.
\end{equation*}

 According to~\cite[Lemma 4.8]{ChenCH021}, general finite POVMs can be simulated using rank-1 POVMs if we only consider the measurement outcomes and disregard the post-measurement quantum state. Thus it is without loss of generality to only consider rank-1 POVMS, i.e., 
\begin{equation}
    M_x^i=\qproj{\psi_x^i}, \quad\sum_{x=1}^k\qproj{\psi_x^i}=\eye_d
    \label{equ:rank1-povm}
\end{equation}
Note that $\qbit{\psi_x^i}$ may not be normalized. 

We can mathematically think of the three kinds of unentangled measurement schemes as follows,

\begin{itemize}
    \item In \textit{fixed} measurement schemes, each $\POVM_i$ is fixed before receiving the quantum state $\rho$.
    \item In \textit{randomized non-adaptive schemes}, there is a common random seed $R\sim \mathcal{R}$ that is generated independently of $\rho$ and the measurements are then chosen as a function of $R$, namely,  $\POVM_i=\POVM_i(R)$. 
    \item For \textit{randomized adaptive schemes}, in addition to the common randomness, the $i$th measurement depends on the previous outcomes, namely, $\POVM_i=\POVM_i(x_1, \ldots, x_{i-1}, R)$.
\end{itemize}

Using fixed measurements, when the state is $\rho$, the $\ns$ outcomes $x_1, \ldots, x_n$ follow a product distribution 
\begin{equation}
\label{equ:outcome-distribution}
    \bfP_\rho \eqdef\bigotimes_{i=1}^n\p^i_{\rho}.
\end{equation}
For randomized non-adaptive measurements, the outcomes are independent conditioned on the random seed $R$, and thus we can write $\bfP_\rho(R)=\bigotimes_{i=1}^n\p^i_{\rho}(R)$. For adaptive measurements, the $\ns$ outcomes are in general not independent anymore. 

We note that without shared randomness, choosing fixed POVMs for each copy is sufficient and additional private randomness is not helpful: suppose that $\POVM_i$ depends on an independent randomness $R_i$, then it is equivalent to choosing a fixed POVM defined by $\{\int M_x^i(R_i) dR_i\}_{x=1}^k$. 

\subsection{Closeness measures of distributions}
Let $\p$ and $\q$ be two distributions over a finite discrete domain $\mathcal{X}$. The \emph{total variation distance} is defined as,
\[
\totalvardist{\p}{\q}\eqdef\sup_{S\subseteq\mathcal{X}}(\p(S)-\q(S))=\frac{1}{2}\sum_{x\in\mathcal{X}}|\p(x)-\q(x)|.
\]
The Kullback-Leibler (KL) divergence $\kldiv{\p}{\q}$ and chi-square divergence $\chisquare{\p}{\q}$ are defined as
\[
\kldiv{\p}{\q}\eqdef\sum_{x\in\mathcal{X}}\p(x)\log\frac{\p(x)}{\q(x)},\quad \chisquare{\p}{\q}\eqdef \sum_{x\in\mathcal{X}}\frac{(\p(x)-\q(x))^2}{\q(x)}.
\]
The three quantities can be related using Pinsker's inequality and concavity of logarithms respectively,
\[
2\totalvardist{\p}{\q}^2\le \kldiv{\p}{\q}\le \chisquare{\p}{\q}.
\]
We may also define $\ell_p$ distances between distributions, $
\norm{\p-\q}_p\eqdef\Paren{\sum_{x\in\mathcal{X}}{|\p(x)-\q(x)|^p}}^{1/p}.
$

\subsection{Operators and quantum channels}

\paragraph{Hilbert space over complex matrices.}
The space of complex matrices $\C^{\dims\times\dims}$ is a Hilbert space when equipped with the matrix inner product defined as $\hdotprod{A}{B}\eqdef\Tr[A^\dagger B]$, where $A, B\in \C^{\dims\times \dims}$. The subspace of all Hermitian matrices, denoted as $\Herm{\dims}$, is a \textit{real} Hilbert space (i.e. the associated field is $\R$) with the same matrix inner product. Any positive semi-definite Hermitian matrix $M$ has a unique p.s.d. square root $K$ such that $K^2=M$, and we denote $K=\sqrt{M}$.

A homomorphism can be defined between $\C^{\dims\times\dims}$ and $\C^{\dims^2}$ through vectorization. On the canonical basis $\{\qbit{j}\}_{j=0}^{\dims-1}$, we define $\VecOp(\qoutprod{i}{j})\eqdef \qbit{j}\otimes \qbit{i}$. Vectorization for general matrices is defined by linearity. Furthermore, the matrix inner product can be equivalently written as the inner product on $\C^{\dims^2}$, $\hdotprod{A}{B}=\VecOp(A)^\dagger\VecOp(B)$.

We can similarly define the tensor product of matrices. Let $A=[a_{ij}]\in \C^{m\times n}$ and $B\in \C^{k\times l}$, then 
$$
A\otimes B=\begin{bmatrix}
    a_{11}B & \ldots & a_{1n}B\\
    \vdots & \ddots & \vdots\\
    a_{m1}B &\ldots & a_{mn}B
\end{bmatrix}\in \C^{mk\times nl}.
$$

\paragraph{(Linear) superoperators and quantum channels.} One can define linear operators over $\C^{\dims\times \dims}$. Since each matrix itself can be viewed as an operator over $\C^\dims$, we refer to linear operators over $\C^{\dims\times \dims}$ as superoperators\footnote{Indeed an operator over $\C^{\dims\times \dims}$ i.e. superoperator need not be linear, but we only deal with linear superoperators in this work, so we drop the word ``linear'' for brevity.} to avoid confusion. Let $\mathcal{N}:\C^{\dims\times \dims}\mapsto \C^{\dims\times \dims}$ be a superoperator. There exists a unique adjoint superoperator $\mathcal{N}^\dagger$ such that for all $X, Y\in \C^{\dims\times \dims}$,
\[
\hdotprod{Y}{\mathcal{N}(X)}=\hdotprod{\mathcal{N}^\dagger(Y)}{X}.
\]

Similar to the trace of matrices, we define the trace of superoperators next.
\begin{definition}
    \label{def:trace-supop}
    Let $\mathcal{N}:\C^{\dims\times \dims}\mapsto \C^{\dims\times \dims}$ be a superoperator. Define its trace as $\Tr[\mathcal{N}]=\sum_{i, j=1}^{\dims}\hdotprod{\qbit{i}\qadjoint{j}}{\mathcal{N}(\qbit{i}\qadjoint{j})}$
\end{definition}

We consider a special superoperator associated with a rank-1 POVM $\POVM=\{M_x\}_{x=1}^k$ where $M_x=\qproj{\psi_x}$. When $\POVM$ acts on $\rho$ and we obtain an outcome $x$, then the equation below gives one possible form of the post-measurement state,
\[
\rho^x\eqdef\frac{K_x\rho K_x}{\Tr[M_x\rho]}, \quad K_x=\sqrt{M_x}=\frac{\qproj{\psi_x}}{\sqrt{\qdotprod{\psi_x}{\psi_x}}}.
\]
This is called the \textit{generalized L\"uder's rule}. We note that the post-measurement states after applying a POVM are undefined in general, and this rule gives just one possible state transformation after measurement.

If we lack the knowledge of measurement outcomes, we can view the underlying state as an expectation of all post-measurement states. This process can be formulated by a superoperator called \textit{L\"uders channel}.
\begin{definition}[L\"uders channel]
\label{def:expected-density}
    Given a POVM $\POVM=\{M_x\}_{x=1}^k$ where $\sum_xM_x=\eye_d$ and $K_x=\sqrt{M_x}$, the L\"uders channel $\Luders_{\POVM}:\C^{\dims\times \dims}\mapsto \C^{\dims\times \dims}$ is defined as
    \begin{equation}
        \Luders_{\POVM}(X) \eqdef \sum_{x=1}^kK_xX K_x = \sum_{x=1}^{k}\frac{\qproj{\psi_x}X\qproj{\psi_x}}{\qdotprod{\psi_x}{\psi_x}}.
        \label{equ:rank-1-luders-channel}
    \end{equation}
    The second equality applies to rank-1 POVM, i.e. $M_x=\qproj{\psi_x}$ for $x=1, \ldots, k$. 
\end{definition}

The L\"uders channel maps a state $\rho$ to an expected post-measurement state $\Luders_{\POVM}(\rho)$.~\cref{def:expected-density} gives its \emph{Kraus representation}. Due to the homomorphism between $\C^{\dims\times\dims}$ and $\C^{\dims^2}$, one can express superoperators using a $\dims^2\times\dims^2$ matrix which is a linear operator over $\C^{\dims^2}$. One such representation is called the\textit{ Choi representation}, e.g.~\cite[Eq (4.3.8)]{khatri2021principles}. For a rank-1 POVM $\POVM$, the Choi representation of $\Luders_{\POVM}$ is
    \begin{equation}
    \label{equ:choi-of-luders}
        \Choi_{\POVM}\eqdef \sum_{x=1}^{k}\frac{\qproj{\bar{\psi}_x}\otimes\qproj{\psi_x}}{\qdotprod{\psi_x}{\psi_x}}.
    \end{equation}
One can easily verify through homomorphism that $\VecOp(\Luders_{\POVM}(X))=\Choi_{\POVM}\VecOp(X)$.

\paragraph{Schatten norms for linear (super)operators} Let $\lambda_1, \ldots, \lambda_\dims\ge 0$ be the \emph{singular values} of a linear operator $A$\footnote{For Hermitian matrices, the singular values are simply the absolute values of the eigenvalues.}, then for $p\ge 1$, the \emph{Schatten $p$-norm} is defined as 
$
\|A\|_{S_p}\eqdef \Paren{\sum_{i=1}^\dims \lambda_i^p}^{1/p}
$, which can be defined for both matrices and superoperators. Some important special cases are trace norm $\tracenorm{A}\eqdef\|A\|_{S_1}$, Hilbert-Schmidt norm $\hsnorm{A}\eqdef\|A\|_{S_2}$, and operator norm $\opnorm{A}\eqdef\|A\|_{S_\infty}=\max_{i=1}^\dims\lambda_i.$. A standard fact is that $\tracenorm{A}=\Tr[\sqrt{A^\dagger A}]$ and $\hsnorm{A}=\sqrt{\hdotprod{A}{A}}$.

\section{Unified lower bound framework for non-adaptive schemes}
\label{sec:lower}
We first state the new lower bound for mixedness testing with fixed measurements in~\cref{thm:no-shared-lb}.
\begin{theorem}
\label{thm:no-shared-lb}
    For $0<\eps<1/200$ and $\dims\ge 16$, without shared randomness, at least $\ns ={\Omega}(\dims^2/\eps^2)$ copies are necessary to test whether $\qs=\qmm$ or $\tracenorm{\qs-\qmm}>\eps$ where $\qmm=\eye_{\dims}/\dims$ is the maximally mixed state.
\end{theorem}

Since mixedness testing is a special case of state certification, this theorem provides a worst-case lower bound for the problem, both when $\qkn$ is known and unknown. 
Recall that $\ns = O(\dims^{3/2}/\eps^2)$ copies are sufficient using randomized schemes. Thus there is a strict separation between algorithms with and without shared randomness, and the gap is a factor of ${\Theta}(\sqrt{\dims})$. 

\cref{thm:no-shared-lb} is an immediate corollary of a unified theoretical framework that we establish for both randomized and fixed non-adaptive measurements. This section will provide a high-level description of the framework.

\subsection{Le Cam's method}

The central tool to prove testing lower bounds is Le Cam's method~\cite{LeCam73,yu1997assouad}. Recall in~\eqref{equ:outcome-distribution} that for a state $\rho$, the distribution of measurement outcomes $\bx=(x_1, \ldots, x_\ns)$ when applying measurements $\POVM=(\POVM_1, \ldots, \POVM_{\ns})$ is $\bfP_{\rho}$. Since we only have access to the measurement outcomes, to distinguish between $\rho$ and $\qmm$, the closeness of outcome distributions $\bfP_{\rho}$ and $\bfP_{\qmm}$ determine the difficulty of testing between the two states. 

For the mixedness testing problem defined in~\cref{sec:preliminaries}, we implicitly require that even if $\rho$ is chosen from any distribution $\mathcal{D}$ over the $\mathcal{P}_\eps\eqdef\{\rho:\tracenorm{\qs-\qmm}>\eps\}$ (the set of states at least $\eps$-far from $\qmm$), we should still be able to distinguish it from the case when the state is $\qmm$. In the former case, the outcome distribution is $\expectDistrOf{\rho\sim\mathcal{D}}{\bfP_{\rho}}$.
Thus, the distance between $\expectDistrOf{\rho\sim\mathcal{D}}{\bfP_{\rho}}$ and $\bfP_{\qmm}$ determines the hardness of mixedness testing. 
More precisely, to guarantee a testing accuracy of at least $2/3$, we need $\frac{2}{3}<\totalvardist{\bfP_{\qmm}}{\expectDistrOf{\sigma\sim\mathcal{D}}{\bfP_{\sigma}}}$.

We note that the distribution $\mathcal{D}$ need not be strictly supported on $\mathcal{P}_\eps$ and only needs to be an \emph{almost-$\eps$ perturbation},  which has constant probability mass over $\mathcal{P}_\eps$.

\begin{definition}[Almost-$\eps$ perturbation]
    A distribution $\mathcal{D}$ is an \emph{almost-$\eps$ perturbation} if $\probaDistrOf{\sigma\sim\mathcal{D}}{\sigma\in\mathcal{P}_\eps}>\frac{1}{2}.$ Denote the set of \emph{almost-$\eps$ perturbation} as $\Gamma_\eps$.
\end{definition}

We show in Lemma~\ref{lem:le-cam-approx-eps-perturbation} that a Le Cam-style lower bound applies to \emph{almost-$\eps$ perturbations}. The proof is in~\cref{app:lem:le-cam-approx-eps-perturbation}.

\begin{lemma}
\label{lem:le-cam-approx-eps-perturbation}
Let $\mathcal{D}$ be an almost-$\eps$ perturbation. Suppose nature either let $Y=0$ and set $\rho=\qmm$ with probability $1/2$, or set $Y=1$ and sample $\rho\sim\mathcal{D}$ with probability $1/2$. Then, if we use a mixedness tester with success probability at least $2/3$ to obtain a guess $\hat{Y}\in\{0, 1\}$, we have $\Pr[Y=\hat{Y}]\ge 1/2$. This implies
   \begin{equation}
        \frac{1}{2}\le \totalvardist{\expectDistrOf{\sigma}{\bfP_{\sigma}}}{\bfP_{\qmm}}\le \sqrt{\frac{1}{2}\chisquare{\expectDistrOf{\sigma}{\bfP_{\sigma}}}{\bfP_{\qmm}}}.
        \label{equ:lecam-total-var}
   \end{equation}
\end{lemma}
Thus to prove copy complexity lower bounds, we need to design an almost-$\eps$ perturbation and then upper bound the chi-square divergence by some function of $\ns, \dims, \eps$.

\subsection{Min-max vs. max-min} Recall in~\cref{sec:fixed-disadvantage} we discussed that the main difference between randomized and fixed measurements is whether nature can choose the hard state adversarially. In this section, we formalize the discussion under a rigorous game theory framework. The testing problem can be viewed as a two-party game played between nature and the algorithm designer, where the algorithm designer tries to design the best algorithms that can distinguish between two states, while nature tries to find hard states to fool the algorithm. 

For a fixed measurement scheme $\POVM^\ns$, nature can choose a $\mathcal{D}\in \Gamma_\eps$ that minimizes the chi-square divergence in~\eqref{equ:lecam-total-var}. According to~\cref{lem:le-cam-approx-eps-perturbation}, if there exists a fixed $\POVM^{\ns}$ that achieves at least 2/3 probability in testing maximally mixed states, we must have
\begin{equation}
\label{equ:max-min-lb}
    \frac{1}{2}\le \max_{\mathcal{M}^n\text{ fixed}}\min_{\mathcal{D}\in \Gamma_\eps}\chisquare{\expectDistrOf{\sigma\sim\mathcal{D}}{\bfP_{\sigma}}}{\bfP_{\qmm}}.
\end{equation}
Thus a max-min game is played between the two parties and nature has an advantage to decide its best action based on the choice of the algorithm designer.

With shared randomness, in principle, a max-min game is still played, but instead, the maximization is over all \textit{distributions} of fixed (non-entangled) measurements.  Using a similar argument as~\cite[Lemma IV.8]{AcharyaCT19}, for the best distribution over all $\mathcal{M}^n$,  the expected accuracy over $R\sim\mathcal{R}$ is at least 1/2 for all $\mathcal{D}\in \Gamma_\eps$. Thus, for all $\mathcal{D}$, there must exist an instantiation $R(\mathcal{D})$ such that using the fixed measurement $\mathcal{M}^n(R(\mathcal{D}))$ the testing accuracy is at least 1/2. Therefore, 
\begin{equation}
\label{equ:min-max-lb}
    \frac{1}{2}\le \min_{\mathcal{D}\in \Gamma_\eps}\max_{\mathcal{M}^n\text{ fixed}}\chisquare{\expectDistrOf{\sigma\sim\mathcal{D}}{\bfP_{\sigma}}}{\bfP_{\qmm}},
\end{equation}
which intuitively says that a min-max game is played and the algorithm designer would have an advantage. 

Therefore, to obtain a sample/copy complexity lower bound for fixed measurements requires upper bounding~\eqref{equ:max-min-lb}, while for randomized schemes requires upper bounding~\eqref{equ:min-max-lb}. We can see that shared randomness is a ``game changer'' that changes a max-min game to a min-max game. Since min-max is no smaller than max-min, testing with shared randomness is easier than testing without it. 

The min-max and max-min arguments in this section are similar to \cite{AcharyaCT19} and we point to \cite[Lemma IV.8, IV.10]{AcharyaCT19} for additional reference.

\subsection{The L\"uders channel characterizes the hardness of testing}
\label{sec:avg-luders}
In the previous section, we give an abstract theoretical framework to prove tight lower bounds for fixed measurements. We now  make it concrete and apply it to mixedness testing. 

Our central contribution is to relate the hardness of testing (i.e., the min-max and max-min divergences) to the \textit{average L\"uders channel} defined by all the POVMs. 
Use the shorthand $\Luders_i\eqdef\Luders_{\POVM_i}$ where $\Luders_{\POVM_i}$ is from~\cref{def:expected-density}, the average  L\"uders channel is defined as
\begin{equation}
    \avgLuders\eqdef\frac{1}{\ns}\sum_{i=1}^{\ns}\Luders_i\;\text{(Kraus)}, \quad\avgChoi\eqdef\frac{1}{\ns}\sum_{i=1}^{\ns}\Choi_i\;\text{(Choi)}.
    \label{equ:average-luders-channel}
\end{equation}

We again use the example in~\cref{sec:fixed-disadvantage} to see why this superoperator is useful. Suppose $\rho$ is the ``+'' state defined in~\eqref{equ:plus-state}. If $\POVM_i=\{\qproj{x}\}_{x=0}^{\dims-1}$, then $\avgLuders(\cdot)=\sum_{x=0}^{\dims-1}\qproj{x}(\cdot)\qproj{x}$. It turns out that $\rho-\qmm$ exactly falls into the 0-eigenspace of $\avgLuders$,
\begin{align*}
    \avgLuders(\rho-\qmm)&=\sum_{x=0}^{d-1}\qproj{x}(\rho-\qmm)\qproj{x}\\
    &=\sum_{x=0}^{d-1}\qbit{x}\qdotprod{x}{\phi}\qdotprod{\phi}{x}\qadjoint{x}-\sum_{x=0}^{d-1}\qproj{x}\frac{\eye_d}{d}\qproj{x}\\
    &=\sum_{x=0}^{d-1}\qbit{x}\qadjoint{x}\frac{1}{d}-\frac{\eye_\dims}{\dims}=0.
\end{align*}

The third equality holds because $\qdotprod{x}{\phi}=\qdotprod{\phi}{x}=1/\sqrt{\dims}$. This serves as an intuitive example that the eigenvalues of $\avgLuders$ superoperator characterize the ability of the measurement scheme to distinguish between quantum states. If the difference $\rho-\qmm$ falls into the eigenspace of $\avgLuders$ with small eigenvalues, then we can expect that the two states are hard to distinguish. 

We can relate the eigenvalues of $\avgLuders$ to the max-min and min-max distances in~\eqref{equ:max-min-lb} and~\eqref{equ:min-max-lb}, thus proving copy complexity lower bounds for both fixed and randomized non-adaptive schemes. We first state some important properties of the eigenvalues and eigenvectors of $\avgLuders$ in~\cref{lem:supop-eigendecomposition}. These are standard results and we state their proofs in~\cref{app:lem:supop-eigendecomposition}.

\begin{restatable}{lemma}{eigendecomp}
\label{lem:supop-eigendecomposition}
    $\avgLuders$ has eigenvectors $\hbasis_{\avgLuders}=(V_1, \ldots, V_{\dims^2})$ with eigenvalues $0\le \lambda_1\le \ldots\le\lambda_{\dims^2}=1$ where
    \begin{enumerate}
        \item $V_i\in\C^{d\times d}$ is a Hermitian matrix with $\Tr[V_i]=0$ for $i\le \dims^2-1$,
        \item $V_{\dims^2}=\eye_\dims/\sqrt{\dims}$.
        \item $\hdotprod{V_i}{V_j}=\Tr[V_iV_j]=\delta_{ij}\eqdef\indic{i=j}$.
    \end{enumerate}
    Thus $\{V_i\}_{i=1}^{\dims^2}$ forms an orthonormal basis for $\C^{\dims\times\dims}$ and $\HH_\dims$, the space of Hermitian matrices. Furthermore, $\Tr[\avgLuders]=\sum_{i=1}^{\dims^2}\lambda_i=\dims$.
\end{restatable}

Our main technical contribution in this paper is~\cref{thm:maxmin-and-minmax-ub}. It establishes upper bounds on the max-min and min-max chi-square divergence in terms of the eigenvalues of $\avgLuders$. Due to the correspondence between the max-min and min-max divergences to the fixed and randomized measurement schemes, we can directly apply the theorem to obtain copy complexity lower bounds for both fixed and randomized schemes as corollaries of this theorem.
\begin{theorem}
\label{thm:maxmin-and-minmax-ub}
    When $\ns = O(\dims^2/\eps^2)$, the max-min chi-square divergence can be bounded as
    \begin{equation}
    \label{equ:maxmin-ub}
        \max_{\mathcal{M}^n\text{ fixed}}\min_{\mathcal{D}\in \Gamma_\eps}\chisquare{\expectDistrOf{\sigma\sim\mathcal{D}}{\bfP_{\sigma}}}{\bfP_{\qmm}}=\bigO{\frac{\ns^2\eps^4}{\dims^4}\cdot\frac{\max_{\avgLuders}\tracenorm{\avgLuders}^2}{\dims^2}},
    \end{equation}
and the min-max chi-square divergence can be bounded as
    \begin{equation}
    \label{equ:minmax-ub}
        \min_{\mathcal{D}\in \Gamma_\eps}\max_{\mathcal{M}^n\text{ fixed}}\chisquare{\expectDistrOf{\sigma\sim\mathcal{D}}{\bfP_{\sigma}}}{\bfP_{\qmm}}=\bigO{\frac{\ns^2\eps^4}{\dims^4}\max_{\avgLuders}\hsnorm{\avgLuders}^2}.
    \end{equation}
\end{theorem}

From~\cref{lem:supop-eigendecomposition}, note that $\tracenorm{\avgLuders}=\dims$. Moreover, $\hsnorm{\avgLuders}^2=\sum_{i=1}^{\dims^2}\lambda_i^2\le (\max_{i}\lambda_i) \sum_{i=1}^{\dims^2}\lambda_i\le \dims$ (since $\lambda_i\le 1$). 
Thus,
\[
\frac{\max_{\avgLuders}\tracenorm{\avgLuders}^2}{\dims^2}=1,\quad\max_{\avgLuders}\hsnorm{\avgLuders}^2\le \dims.
\]
Combining~\eqref{equ:max-min-lb} and~\eqref{equ:maxmin-ub}, we conclude that for fixed measurements $\ns=\Omega(\dims^2/\eps^2)$ and prove~\cref{thm:no-shared-lb}. Combining~\eqref{equ:min-max-lb} and~\eqref{equ:minmax-ub}, we recover the $n=\Omega(\dims^{3/2}/\eps^2)$ lower bound for randomized non-adaptive schemes, which was shown in~\cite{BubeckC020}.

\begin{remark}
    One can show that $\avgLuders$ is the L\"uders channel of a POVM $\POVM\eqdef\{\frac{1}{\ns}M_x^i\}_{x\in[k], i\in[\ns]}$ which is the ensemble of all measurements. One can define $\POVM^\dagger \POVM$ where we slightly abused the notation and treated $\POVM:\Herm{\dims}\mapsto \R^k$ as a linear mapping from quantum states to probability vectors. $\avgLuders$ and $\POVM^\dagger \POVM$ are similar but slightly different superoperators\footnote{The differ by a scalar factor if $\qbit{\psi_x^i}$ have equal norms, but can be very different otherwise.}.~\cite{guctua2020fast} used $\POVM^\dagger\POVM$ to derive \textit{upper bounds} for quantum tomography for three specific types of measurements. Our result is orthogonal to their work in that we prove lower bounds for general rank-1 measurements.
\end{remark}

\section{Proof of~\cref{thm:maxmin-and-minmax-ub}}
We now prove~\cref{thm:maxmin-and-minmax-ub}, unifying the lower bounds for randomized and fixed non-adaptive measurements. We start with a generic construction of hard case quantum states which allows the lower bound instance to be measurement-dependent. Then we formally relate the chi-square divergence in~\eqref{equ:lecam-total-var} to the L\"uders channel. With these ingredients, we finally show that the min-max and max-min chi-square divergences are characterized by different eigenvalues of the L\"uders channel, thus completing the proof.

\subsection{Hard instance construction}
We now describe a generic construction of a distribution over density matrices that will serve as hard case for both fixed and randomized measurements. {We will take a finite set of trace-0 orthonormal Hermitian matrices. Then, we take a linear combination of these matrices with coefficients chosen at random to be $\pm 1$ (appropriately normalized). When we add maximally mixed state $\qmm$ to this distribution's output we obtain a perturbed distribution around $\qmm$.}
 \begin{definition}
 \label{def:perturbation}
     Let $\frac{\dims^2}{2}\le \ell\le \dims^2-1$ and $\hbasis=(V_1, \ldots, V_{\dims^2})$ be an orthonormal basis of $\Herm{\dims}$ with $V_{\dims^2}=\eye_\dims/\sqrt{\dims}.$ Define $\ptbDistr(\hbasis)$ as follows. Let  $z=(z_1, \ldots, z_\ell)\in\{-1, 1\}^\ell$ be uniformly drawn from the $\{-1, 1\}^\ell$ hypercube. Let $\cd$ be a universal constant, then define
     \begin{equation}
         \Delta_z = \frac{c\eps}{\sqrt{\dims}}\cdot\frac{1}{\sqrt{\ell}}\Paren{\sum_{i=1}^\ell z_iV_i}, \quad \barDelta_z= \Delta_z\min\left\{1, \frac{1}{\dims \opnorm{\Delta_z}}\right\}.
         \label{equ:delta_z}
     \end{equation}
     $\barDelta_z$ normalizes $\Delta_z$ so that its maximum absolute eigenvalue is at most $1/\dims$. Let
     \begin{equation}
         \sigma_z=\qmm + \barDelta_z,\quad 
         \label{equ:sigma_z}
     \end{equation}
     This defines a distribution over states (induced by the randomness in $z$) which we denote by $\ptbDistr(\hbasis)$.
 \end{definition}

Note that after normalizing the perturbations, $\sigma_z$ is a valid density matrix. 
However, the trace distance $\|\sigma_z-\qmm\|_1$ may not be greater than $\eps$. 
Nevertheless, we can show that the probability of this bad event is negligible. The central claim is ~\cref{thm:rand-mat-opnorm-concentration} which states that the operator norm of a random matrix with independently perturbed orthogonal components is $O(\sqrt{\dims})$ with high probability. The proof is in~\cref{app:prop:perturbation-trace-distance}.
\begin{restatable}{theorem}{randmatopnorm}
\label{thm:rand-mat-opnorm-concentration}
    Let $V_1, \ldots, V_{\dims^2}\in\C^{\dims\times \dims}$ be an orthonormal basis of $\C^{\dims\times \dims}$ and $\ptb_1, \ldots, \ptb_{\dims^2}\in\{-1, 1\}$ be independent symmetric Bernoulli random variables. Let $W=\sum_{i=1}^{\ell}\ptb_iV_i$ where $\ell\le \dims^2$. For all $\alpha>0$, there exists $\cop_\alpha$, {which is increasing in $\alpha$} such that
    \[
    \probaOf{\opnorm{W}>\cop_\alpha\sqrt{\dims}}\le 2\exp\{-\alpha\dims\}.
    \]
\end{restatable}
\begin{remark}
    Standard random matrix theory (e.g. \cite{tao2023topics}[Corollary 2.3.5]) states that if each entry of $W$ is independent and uniform from $\{-1, 1\}$, i.e. $W=\sum_{i,j}z_{ij}E_{ij}$ where $E_{ij}$ is a matrix with 1 at position $(i, j)$ and 0 everywhere else, then $\opnorm{W}=O(\sqrt{\dims})$ with high probability.~\cref{thm:rand-mat-opnorm-concentration} generalizes this argument to arbitrary basis $\{V_i\}_{i=1}^{\dims^2}$. {This could be of independent interest.}
\end{remark}

An immediate corollary of~\cref{thm:rand-mat-opnorm-concentration} is that with appropriately chosen constant $c$ in~\cref{def:perturbation}, $\ptbDistr(\hbasis)$ is an almost-$\eps$ perturbation and~\cref{lem:le-cam-approx-eps-perturbation} can be applied to $\ptbDistr(\hbasis)$. 
\begin{corollary}
\label{prop:perturbation-trace-distance}
    Let $\dims ^2/2\le \ell\le \dims^2-1$. Let  $z$ be drawn from a uniform distribution over $\{-1,1\}^{\ell}$ , and $\Delta_z, \sigma_z$ are as defined in~\cref{def:perturbation}. Then, there exists a universal constant $\cd\le 10\sqrt{2}$, such that for $\eps<\frac{1}{\cd^2}$, with probability at least $1-2\exp(-\dims)$, $\opnorm{\Delta_z}\le 1/\dims$ and $\|\Delta_z\|_1\ge \eps$.   
\end{corollary}
\begin{proof}
By H\"older's inequality, we have that for all matrices $A$,
\[
\opnorm{A}\tracenorm{A}\ge \hsnorm{A}^2.
\]
Note that $\Delta_z=\frac{\cd\eps}{\sqrt{\dims\ell}}W$ and $\hsnorm{\Delta_z}=\frac{c\eps}{\sqrt{\dims}}$. Thus setting $\alpha=1$ and $\cop=\cop_1$ in~\cref{thm:rand-mat-opnorm-concentration}, with probability at least $1-2\exp(-\dims)$,
\[
\opnorm{\Delta_z}\le \frac{\cd\eps}{\sqrt{\dims\ell}}\cdot \kappa \sqrt{\dims}=\frac{c\kappa\eps}{\sqrt{\ell}}.
\]
 This implies \newja{that}
\[
\tracenorm{\Delta_z}\ge \hsnorm{\Delta_z}^2/\opnorm{\Delta_z}\ge \frac{c\eps}{\kappa}\cdot\frac{\sqrt{\ell}}{\dims}.
\]
In the proof of~\cref{thm:rand-mat-opnorm-concentration} in~\cref{app:prop:perturbation-trace-distance}, we can show that $\kappa=\kappa_1\le 10$. Thus choosing $\cd=\sqrt{2}\kappa\le 10\sqrt{2}$, we guarantee that $\tracenorm{\Delta_z}>\eps$ due to $\ell\ge \dims^2/2$. As long as $\eps\le \frac{1}{200}$, we have $\opnorm{\Delta_z}\le 1/\dims$ and thus $\sigma_z=\qmm + \Delta_z$ is a valid density matrix. This completes the proof of~\cref{prop:perturbation-trace-distance}.
\end{proof}

Different bounds for min-max and max-min divergences in~\cref{thm:maxmin-and-minmax-ub} are due to whether or not nature can choose $\hbasis$ dependent on $\avgLuders$, which in turn depends on the measurements $\POVM^\ns$. For randomized schemes, we need to upper bound the min-max divergence, and we can simply choose a fixed $\hbasis$ that is uniformly bad for all $\POVM^\ns$. For fixed measurements however, under the max-min framework, nature could choose the hard distribution depending on $\POVM^\ns$. Specifically, with $\hbasis=\hbasis_{\avgLuders}$ and $\ell$ small, $\sigma_z-\qmm$ completely lies in an eigenspace of $\avgLuders$ with the $\ell$ smallest eigenvalues, thus generalizing the intuition from the toy example in~\cref{sec:fixed-disadvantage}.

\subsection{Relating chi-squared divergence to L\"uders channel}
We would like to bound the chi-squared distance between the outputs of the measurements in the cases when the input is the maximally mixed state, versus the case when it is chosen from an $\eps$ perturbation. In particular, we would like to bound~\eqref{equ:lecam-total-var}. We will bound this in terms of the average L\"uders channel $\avgLuders$ in~\cref{lem:decoupled-chi-square}. 
This will be our main technical lemma, which will formalize the intuition established in~\cref{sec:avg-luders} that for fixed measurements we should try to choose $\Delta_\sigma$ from a subspace with small eigenvalues to {obtain strong} lower bounds. Due to its importance, we provide the complete proof.
\begin{lemma}
\label{lem:decoupled-chi-square}
Let $\sigma, \sigma'$ be independently drawn from a distribution $\mathcal{D}$, and $\POVM_i$ be rank-1 POVM as in~\eqref{equ:rank1-povm} for $i=1, \ldots, n$. Define $\Delta_\sigma = \sigma -\qmm$. Then
    \begin{align}
    \label{equ:chi-square-quantum}
        \chisquare{\expectDistrOf{\sigma\sim\mathcal{D}}{\bfP_{\sigma}}}{\bfP_{\qmm}}&\le \expectDistrOf{\sigma, \sigma'\sim \mathcal{D}}{\exp\left\{\ns\dims \hdotprod{\Delta_{\sigma}}{\avgLuders(\Delta_{\sigma'})}\right\}}-1\\
        &=\expectDistrOf{\sigma, \sigma'\sim \mathcal{D}}{\exp\left\{\ns\dims \cdot\VecOp(\Delta_\sigma)^\dagger\avgChoi\VecOp(\Delta_{\sigma'}) \right\}}-1\nonumber
    \end{align}
    where $\avgLuders$ is the average L\"uders channel defined in~\cref{equ:average-luders-channel} and $\avgChoi$ is its Choi representation.
\end{lemma}
\begin{proof} 
    The proof uses ideas from the decoupled chi-square fluctuations introduced in \cite{AcharyaCT19}.  We can directly bound the chi-square distance using the following lemma which is from~\cite{Pollard:2003}.
    \begin{lemma}[{\cite{Pollard:2003},\cite[Lemma III.8]{AcharyaCT19}}]
    \label{lem:chi-square-expansion}
        Let $\bfP=\p\supparen{1}\otimes \cdots\otimes \p\supparen{n}$ be a fixed product distribution and $\bfQ_\theta=\q_{\theta}\supparen{1}\otimes \cdots \otimes \q_{\theta}\supparen{n}$ be parameterized by a random variable $\theta$. Then
        \[
        \chisquare{\expectDistrOf{\theta}{\bfQ_{\theta}}}{\bfP}=\expectDistrOf{\theta, \theta'}{\prod_{i=1}^n(1+H_i(\theta, \theta'))}-1,
        \]
        where $\theta'$ is an independent copy of $\theta$ and 
        \[
        H_i(\theta, \theta')\eqdef \expectDistrOf{x\sim \p\supparen{i}}{\delta_\theta\supparen{i}(x)\delta_{\theta'}\supparen{i}(x)},\quad \delta_{\theta}\supparen{i}(x)\eqdef\frac{\q_\theta\supparen{i}(x)-\p\supparen{i}(x)}{\p\supparen{i}(x)}.
        \]
    \end{lemma}

In our problem, $\bfP$ will be $\bfP_{\qmm}$, the distribution over the output of measurements across the $n$ copies when the underlying state is  maximally mixed, and $\expectDistrOf{\theta}{\bfQ_{\theta}}$ will be $\bfP_{\sigma}$, the mixture distribution over the output of measurements when the underlying state is  parameterized by a random density matrix $\sigma$ induced by the perturbation. These are defined in~\eqref{equ:outcome-distribution}.

We first compute the necessary quantities by appropriate substitution. Recall that $\p_{\rho}^i(\cdot)$ is the output distribution of the measurement on the $i$th copy. 
    \[
    \delta_{\sigma}^i(x)=\frac{\p_{\sigma}^i(x)-\p_{\qmm}^i(x)}{\p_{\qmm}^i(x)}, \;x\in [k].
    \]
    We now evaluate $H_i(\sigma, \sigma')$ by expanding the probabilities using Born's rule.
    \begin{align*}
        H_i(\sigma, \sigma') &=\expectDistrOf{x\sim\p_{\qmm}^i}{\frac{(\p_{\sigma}^i(x)-\p_{\qmm}^i(x))(\p_{\sigma'}^i(x)-\p_{\qmm}^i(x))}{(\p_{\qmm}^i(x))^2}}\\
        &=\sum_{x}\frac{(\p_{\sigma}^i(x)-\p_{\qmm}^i(x))(\p_{\sigma'}^i(x)-\p_{\qmm}^i(x))}{\p_{\qmm}^i(x)}\\
        &=\sum_{x}\frac{\matdotprod{\psi_x^i}{\Delta_{\sigma}}{\psi_x^i}\matdotprod{\psi_x^i}{\Delta_{\sigma'}}{\psi_x^i}}{\qdotprod{\psi_x^i}{\psi_x^i}/\dims}.
    \end{align*}
    This expression can now be related to the L\"uders channel. Adding trace to the numerator does not change the value, and from this we can apply cyclicity and linearity of trace,
    \begin{align*}
        H_i(\sigma, \sigma')&=\dims \Paren{\sum_{x}\frac{\Tr[{\Delta_{\sigma}}\qbit{\psi_x^i}\matdotprod{\psi_x^i}{\Delta_{\sigma'}}{\psi_x^i}\qadjoint{\psi_x^i}]}{\qdotprod{\psi_x^i}{\psi_x^i}}}\\
        &=\dims\cdot \Tr\left[\sum_{x}\frac{{\Delta_{\sigma}}\qbit{\psi_x^i}\matdotprod{\psi_x^i}{\Delta_{\sigma'}}{\psi_x^i}\qadjoint{\psi_x^i}}{\qdotprod{\psi_x^i}{\psi_x^i}}\right]\\
        &=\dims\cdot \Tr\left[{\Delta_{\sigma}}\sum_{x}\frac{\qbit{\psi_x^i}\matdotprod{\psi_x^i}{\Delta_{\sigma'}}{\psi_x^i}\qadjoint{\psi_x^i}}{\qdotprod{\psi_x^i}{\psi_x^i}}\right]\\
        &=\dims\cdot \Tr\left[\Delta_{\sigma}\Luders_i(\Delta_{\sigma'})\right]
        =\dims \hdotprod{\Delta_{\sigma}}{\Luders_i(\Delta_{\sigma'})}\in \R,
    \end{align*}
where the last step uses the fact that $\Delta_{\sigma}$ is Hermitian.

Then, using~\cref{lem:chi-square-expansion}, and the fact that $1+x\le \exp(x)$, we obtain
    \begin{align*}
        \chisquare{\expectDistrOf{\sigma\sim\mathcal{D}}{\bfP_{\sigma}}}{\bfP_{\qmm}}&=\expectDistrOf{\sigma, \sigma'}{\prod_{i=1}^n(1+H_i(\sigma, \sigma'))}-1\\
        &\le \expectDistrOf{\sigma,\sigma'}{\exp\left\{\sum_{i=1}^nH_i(\sigma, \sigma')\right\}}-1\\
        &=\expectDistrOf{\sigma,\sigma'}{\exp\left\{\dims\sum_{i=1}^\ns\hdotprod{\Delta_{\sigma}}{\Luders_i(\Delta_{\sigma})}\right\}}-1.
    \end{align*}
    By linearity of the Hibert-Schmidt inner product and definition of $\avgLuders$,
    \begin{align*}
    \chisquare{\expectDistrOf{\sigma\sim\mathcal{D}}{\bfP_{\sigma}}}{\bfP_{\qmm}}
       &\le \expectDistrOf{\sigma,\sigma'}{\exp\left\{\ns\dims\hdotprod{\Delta_{\sigma'}}{\frac{1}{\ns}\sum_{i=1}^{\ns}\Luders_i(\Delta_\sigma)}\right\}}-1\\
        &=\expectDistrOf{\sigma, \sigma'}{\exp\{\ns\dims\hdotprod{\Delta_{\sigma}}{\bar{\mathcal{H}}(\Delta_{\sigma'})}\}}-1.
    \end{align*}
    Using homomorphism $\VecOp(\Luders_{\POVM}(X))=\Choi_{\POVM}\VecOp(X)$, we have $\hdotprod{\Delta_{\sigma}}{\bar{\mathcal{H}}(\Delta_{\sigma'})}=\VecOp(\Delta_{\sigma})\avgChoi\VecOp(\Delta_{\sigma'})$, completing the proof.
\end{proof}

\paragraph{Explaining the example in~\cref{sec:fixed-disadvantage}.} We now use~\cref{lem:decoupled-chi-square} to explain why choosing a fixed basis measurement $\{\qproj{x}\}_{x=0}^{\dims-1}$ for all copies as in~\cref{sec:fixed-disadvantage} would fail. Since there are only $\dims$ rank-1 projectors, the rank of $\avgChoi$ is $\dims$, but $\avgChoi$ has a dimension of $\dims^2\times\dims^2$ and thus there are a total of $\dims^2-\dims$ eigenvectors with 0 eigenvalues. From~\cref{prop:perturbation-trace-distance}, we know that there must exist a trace-0 $\Delta$ in the 0-eigenspace such that $\sigma=\qmm+\Delta\in\mathcal{P}_{\eps}$. For this particular $\sigma$ the upper bound in~\eqref{equ:chi-square-quantum} is 0, and thus it is impossible to distinguish $\qmm$ and $\sigma$. This is consistent with the discussion in~\cref{sec:fixed-disadvantage}.

We can make a more general argument that to avoid the catastrophic failure similar to the dummy example in~\cref{sec:fixed-disadvantage}, $\avgChoi$ has to be nearly full-rank: $\text{rank}(\avgChoi)\ge (1-o(1))\dims^2$. Thus $(1-o(1))\dims^2$ linearly independent rank-1 projectors are needed in all the POVMs. Indeed if otherwise, the dimension of the 0-eigenspace of $\avgLuders$ is $\Omega(\dims^2)$, we can again invoke~\cref{prop:perturbation-trace-distance} (perhaps with some different constants) to find a \emph{single fixed} $\sigma$ that completely fools the measurement scheme.

\subsection{Bounding the max-min and min-max divergences}
\label{sec:lower-proof}

Let $\hbasis=(V_1, \ldots, V_{\dims^2}=\eye_\dims/\sqrt{\dims})$ be an orthonormal basis of $\Herm{\dims}$. We now upper bound the expression~\eqref{equ:chi-square-quantum} in~\cref{lem:decoupled-chi-square} when $\mathcal{D}=\ptbDistr(\hbasis)$, defined in~\cref{def:perturbation}. The result is in~\cref{thm:chi-square-upper-bound}. 
The central claim is that the chi-squared divergence is related to the Hilbert-Schmidt norm of the projection of $\avgLuders$ onto the subspace defined by $V_1, \ldots, V_{\ell}$.

\begin{restatable}{theorem}{chisqub}
\label{thm:chi-square-upper-bound}
    Let $\frac{\dims^2}{2}\le \ell\le \dims^2-1$, $\hbasis=(V_1, \ldots, V_{\dims^2}=\eye_\dims/\sqrt{\dims})$ be an orthonormal basis of $\Herm{\dims}$, $ V\eqdef[\VecOp(V_1), \ldots, \VecOp(V_\ell)]$ and $\sigma_z, \sigma_{z'}\sim \ptbDistr(\hbasis)$ defined in~\cref{def:perturbation}. Then for $\ns<\frac{\dims^2}{6\cd^2\eps^2}$,
    \begin{equation}
        \expectDistrOf{\sigma_z,\sigma_{z'}}{\exp\left\{\ns\dims \hdotprod{\barDelta_{z'}}{\avgLuders(\barDelta_z)}\right\}}-1\le \exp\left\{\frac{\cd^2\ns^2\eps^4}{\ell^2}\hsnorm{V^\dagger \avgChoi V}^2\right\}-1+\frac{4}{e^\dims}.
        \label{equ:chi-square-final-ub}
    \end{equation} 
\end{restatable}

We now bound $\hsnorm{V^\dagger \avgChoi V}^2$, which depends on how the basis $\hbasis$ is chosen.
\begin{restatable}{observation}{hsnnorm}
\label{obs:hs-norm-luders-projection}
    For all orthonormal basis $\hbasis$, we have $\hsnorm{V^\dagger \avgChoi V}\le \hsnorm{\avgLuders}$. However when $\hbasis=\hbasis_{\avgLuders}$ in~\cref{lem:supop-eigendecomposition}, for all $\frac{\dims^2}{2}\le\ell\le \dims^2-1$, $\hsnorm{V^\dagger \avgChoi V}=\hsnorm{\avgLuders_\ell}\eqdef\sqrt{\sum_{i=1}^{\ell}\lambda_i^2}$ and $0\le \lambda_1\le\ldots\le\lambda_{\dims^2}=1$ are the eigenvalues of $\avgLuders$.
\end{restatable}

The proof of ~\cref{thm:chi-square-upper-bound} is in~\cref{app:thm:chi-square-upper-bound} and the proof of~\cref{obs:hs-norm-luders-projection} is in~\cref{app:obs:hs-norm-luders-projection}. We can now prove Theorem~\ref{thm:maxmin-and-minmax-ub}. It is more straightforward to prove the min-max upper bound~\eqref{equ:min-max-lb} by setting $\hbasis$ as an arbitrary fixed basis that satisfies~\cref{def:perturbation}. For example, one can choose the generalized Gell-Mann basis, 
\begin{align*}
            \sigma_{0,0}&\eqdef\frac{\eye_\dims}{\sqrt{\dims}},\\
            \sigma_{k,l}\supparen{+}&\eqdef \frac{1}{\sqrt{2}}(\qoutprod{k}{l}+\qoutprod{l}{k}), &0\le k<l\le \dims-1,\\
            \sigma_{k,l}\supparen{i}&\eqdef \frac{1}{\sqrt{2}}(-i\qoutprod{k}{l}+i\qoutprod{l}{k}), &0\le k<l\le \dims-1,\\
            \sigma_{k,k}&\eqdef \frac{k}{k+1}\Paren{-k\qproj{k}+\sum_{j=0}^{k-1}\qproj{j}},&1\le k\le \dims-1.
        \end{align*}
We can relabel them as $V_1, \ldots, V_{\dims^2}$ where $V_{\dims^2}=\sigma_{0,0}$. This is a natural extension of Pauli matrices for $\dims=2$. It can be easily verified that these $\dims^2$ matrices indeed form an orthonormal basis over $\Herm{\dims}$. Using~\cref{lem:decoupled-chi-square} and~\cref{thm:chi-square-upper-bound}, setting $\ell=\dims^2-1$,
\begin{align*}
    \min_{\mathcal{D}\in \Gamma_\eps}\max_{\mathcal{M}^n\text{ fixed}}\chisquare{\expectDistrOf{\sigma\sim\mathcal{D}}{\bfP_{\sigma}}}{\bfP_{\qmm}}&\le \max_{\mathcal{M}^n\text{ fixed}}\chisquare{\expectDistrOf{\sigma\sim\ptbDistr(\hbasis)}{\bfP_{\sigma}}}{\bfP_{\qmm}}\\
    &\le \bigO{\frac{\ns^2\eps^4}{\dims^4}\max_{\avgLuders}\hsnorm{\avgLuders}^2}.
\end{align*}

When upper bounding the max-min divergence~\eqref{equ:maxmin-ub}, we would have the freedom to choose a basis $\hbasis$ that depends on $\avgLuders$, 
which is determined by the measurement $\POVM^\ns$. More precisely, we can set $\hbasis=\hbasis_{\avgLuders}$ and $\ell=\dims^2/2$, 
and the perturbations $\ptbDistr(\hbasis_{\avgLuders})$ would be along directions that are least sensitive for the measurement scheme, which leads to the extra $\dims$ factor in the chi-square divergence upper bound,
\begin{align*}
    \max_{\mathcal{M}^n\text{ fixed}}\min_{\mathcal{D}\in \Gamma_\eps}\chisquare{\expectDistrOf{\sigma\sim\mathcal{D}}{\bfP_{\sigma}}}{\bfP_{\qmm}}&\le \max_{\mathcal{M}^n\text{ fixed}}\chisquare{\expectDistrOf{\sigma\sim\ptbDistr(\hbasis_{\avgLuders})}{\bfP_{\sigma}}}{\bfP_{\qmm}}\\
    &\le \bigO{\frac{\ns^2\eps^4}{\dims^4}\max_{\avgLuders}\hsnorm{\avgLuders_\ell}^2}.
\end{align*}
The square-sum of the smallest eigenvalues can be bounded in terms of $\Tr[\avgLuders]$,
\begin{align*}
    \hsnorm{\avgLuders_\ell}^2=\sum_{i=1}^{\ell}\lambda_i^2\le \ell\lambda_\ell^2\le \ell\Paren{\frac{\Tr[\avgLuders]}{\dims^2-\ell}}^2=\frac{2\tracenorm{\avgLuders}^2}{\dims^2}=2.
\end{align*}
The second inequality is because all eigenvalues are sorted in increasing order, and thus $\lambda_\ell$ is no greater than the average of $\lambda_{\ell+1}, \ldots, \lambda_{\dims^2}$, which is at most $\Tr[\avgLuders]/(\dims^2-\ell)$. The proof is complete.


\section{Upper bound for fixed measurements}
\label{sec:upper-bound}

The algorithm we present is similar to an algorithm proposed in~\cite[Algorithm 4]{Yu21sample}\footnote{We came across the result after writing a draft of the paper. However, given the similarity of the algorithms, it should be attributed to~\cite{Yu21sample}.}. They specifically work with maximal mutually unbiased bases~\cite{klappenecker2005mutually}, and we work with quantum $2$-designs, which are generalizations of the former. For completeness, we present the algorithm and its copy complexity guarantee.

The algorithm is based on quantum 2-designs, a finite set of vectors that preserves the second moment of the Haar measure and yields a rank-1 POVM with appropriate scaling. The same measurement is applied to \emph{all} copies. Since it preserves the statistics of the Haar measure, one can show that when $\rho$ and $\qkn$ are far, then the outcome distribution on each copy is also far in $\ell_2$ distance. From this, we apply classical closeness testing to the outcomes. As long as the 2-design has size at most $O(\dims^2)$, then we can achieve the desired $O(\dims^2/\eps^2)$ copy complexity. For $\dims$ that are prime powers, such 2-design exists due to maximal mutually unbiased bases~\cite{klappenecker2005mutually}. This is already general enough since the system dimension $\dims$ is $2^N$ for quantum computers implemented in qubits. Moreover, the algorithm can be easily generalized to the problem of \emph{closedness} testing, where the goal is to test whether two \emph{unknown} states $\rho$ and $\sigma$ are close in trace distance given $\ns$ copies from each.

\subsection{Preliminaries}
\paragraph{Quantum $t$-designs.} At a high level, for an integer $t>0$, $t$-design is a finite set of unit vectors such that the average of any polynomial $f$ of degree at most $t$ is the same as the expectation of $f$ over the Haar measure.
\begin{definition}[Quantum $t$-design]
\label{def:spherical-t-design}
    Let $t$ be a positive integer, we say that a finite set of normalized vectors $\{\qbit{\psi_x}\}_{x=1}^k$ in $\C^\dims$ and a discrete distribution $q=(q_1, \ldots, q_k)$ over $[k]$ a \textit{quantum $t$-design} if
    \[
    \sum_{x=1}^k q_x\qproj{\psi_x}^{\otimes t}=\int \qproj{\psi}^{\otimes t} d \mu(\psi),
    \]
    where $\mu$ is the Haar measure on the unit sphere in $\C^\dims$. If $q_x=1/k$, then the $t$-design is \textit{proper} and we may omit the distribution $q$ when describing proper $t$-designs.
\end{definition}
By taking the partial trace on both sides, we can easily see that a $t$-design is naturally a $t'$-design for all $t'\le t$. Moreover, when $t=1$, the right-hand side is $\eye_\dims/\dims$ and thus $\{\dims q_x\qproj{\psi_x}\}$ is a POVM. An important example of spherical 2-design is based on mutually unbiased bases (MUB) (see \cite{durt2010mutually} for a survey).



\begin{theorem}[\cite{klappenecker2005mutually}]
    Let $\dims$ be a prime power, then there exists a maximal MUB, i.e. $\dims+1$ orthonormal bases $\{\qbit{\psi_{x}^l}\}_{x=1}^{\dims}, l=1, \ldots, \dims+1$ such that the collection of all vectors $\{\qbit{\psi_{x}^l}\}_{x, l}$ is a proper $2$-design.
    \label{thm:mub-2-design}
\end{theorem}

\paragraph{Classical distribution testing} We will use the classical closeness testing algorithm for discrete distributions as a sub-routine. Given two distributions $\p$ and $\q$ and $\ns$ samples from each, the goal is to test whether $\p=\q$ or $\normtwo{\p-\q}\ge \eps$. The sample complexity guarantee is given by the following theorem.
\begin{theorem}[{\cite[Lemma 2.3]{DiakonikolasK16},\cite[Proposition 3.1]{ChanDVV14}}]
     \label{thm:distr-testing}
     Let $\p,\q$ be unknown distributions over $k$ such that $\min\{\normtwo{p}, \normtwo{q}\}\le b$. There exists an algorithm TestClosenessL2($\bx,\bx',\eps$) that distinguishes whether $\p=\q$ or $\normtwo{\p-\q}>\eps$, where $\bx$ and $\bx'$ are $O(b/\eps^2)$ samples from $\p$ and $\q$ respectively.
\end{theorem}


\subsection{Algorithm}
The algorithm applies a proper 2-design for all copies, with suitable coefficients to make the projection matrices a POVM. 2-designs preserve the statistics of the Haar measure up to order 2, and therefore should be a good choice for fixed measurements.
\begin{algorithm}
\caption{State certification/closedness testing without shared randomness}
\label{alg:no-shared}
    \begin{algorithmic}
        \State \textbf{Input}: $\ns$ copies of unknown state $\rho$. If $\qkn$ is unknown, $\ns$ copies of $\qkn$ as well.
        \State \textbf{Output} YES if $\rho=\qkn$, NO if $\tracenorm{\rho-\qkn}>\eps$. 
        \State Let $\{\qbit{\psi_x}\}_{x=1}^k$ be a proper 2-design and $\POVM = \{\frac{\dims}{k}\qproj{\psi_x}\}_{x=1}^k$
        \State Apply the measurement $\POVM$ for all copies of $\rho$ and obtain outcomes $\bx=(x_1, \ldots, x_\ns)$.
        \State Obtain $\ns$ samples $\bx'=(x_1', \ldots, x_\ns')$ from $\p_{\qkn}$. If $\qkn$ is known, then $\bx'$ is obtained by measuring each copy with $\POVM$. Else, $\bx'$ is sampled using classical randomness.
        \State \Return TestClosenessL2$(\bx, \bx', \eps/\sqrt{k(\dims+1)})$.
    \end{algorithmic}
\end{algorithm}
\begin{theorem}
\label{thm:upper}
    Let $k$ be the size of the proper 2-design used in~\cref{alg:no-shared}. With $\ns=\bigO{\dims\sqrt{k}/\eps^2}$ copies from each unknown state, Algorithm~\ref{alg:no-shared} can test whether $\rho=\qkn$ or $\tracenorm{\rho-\qkn}>\eps$ with probability at least $2/3$. 
\end{theorem}
\begin{proof}
    Let $\Delta=\rho-\qkn$ and $\p_{\rho}$ be the distribution of a single measurement outcome for $\POVM$. When $\tracenorm{\rho-\qkn}=\tracenorm{\Delta}\ge \eps$, we have $\hsnorm{\Delta}=\sqrt{\Tr[\Delta^2]}\ge \eps/\sqrt{\dims}$.
    
    We can compute the $\normtwo{\p_{\rho}}$ and $\normtwo{\p_{\rho}-\p_{\qkn}}$ in terms of $\Delta$.
    \[
    \normtwo{\p_{\rho}}^2=\frac{\dims^2}{k^2}\sum_{x=1}^k\matdotprod{\psi_x}{\rho}{\psi_x},\quad \normtwo{\p_{\rho}-\p_{\qmm}}^2=\frac{\dims^2}{k^2}\sum_{x=1}^k\matdotprod{\psi_x}{\Delta}{\psi_x}^2.
    \]
   Note that $\{\qbit{\psi_x}\}_{x=1}^k$ is a proper 2-design, and thus by definition for all Hermitian matrices $M$,
   \begin{align*}
   \frac{1}{k}\sum_{x=1}^k\matdotprod{\psi_x}{M}{\psi_x}^2&=\frac{1}{k}\sum_{x}{\Tr[\qproj{\psi_x}M]^2}=\Tr\left[\frac{1}{k}\sum_{x}\qproj{\psi_x}^{\otimes2}M^{\otimes 2}\right]\\
         &=\Tr\left[\expectDistrOf{\psi\sim\mu}{\qproj{\psi}^{\otimes 2}]M^{\otimes 2}}\right]=\expectDistrOf{\psi\sim\mu}{\Tr[\qproj{\psi}^{\otimes 2}M^{\otimes 2}]}\\
         &=\expectDistrOf{\psi\sim\mu}{\Tr[\qproj{\psi}M]^2}\\
         &=\expectDistrOf{\psi\sim\mu}{\matdotprod{\psi}{M}{\psi}^2},
   \end{align*}
    where $\mu$ is the Haar measure. The expectation can be computed using Weingarten calculus~\cite{collins2003moments,collins2006integration}.
    \begin{lemma}
    \label{lem:haar-square-expectation}
        For any Hermitian $M\in \C^{\dims\times\dims}$ and $\qbit{\psi}\sim \mu$ the Haar measure, we have,
        \[
        \expectDistrOf{\psi\sim\mu}{\matdotprod{\psi}{M}{\psi}^2}=\frac{1}{\dims(\dims+1)}(\Tr[M]^2+\Tr[M^2]).
        \]
    \end{lemma}
    The proof can be found in~\cite[Lemma 6.4]{ChenLO22instance}. 
    Since $\Tr[\rho^2]\le \Tr[\rho]=1$ and $\Tr[\Delta]=0$, from this lemma we conclude that  
    \[
    \normtwo{\p_{\rho}}^2\le\frac{2\dims}{k(\dims+1)},\quad \normtwo{\p_{\rho}-\p_{\qmm}}^2=\frac{
    \dims^2\Tr[\Delta^2]}{k\dims(\dims+1)}\ge\frac{\eps^2}{k(\dims+1)}.
    \]
    Therefore, we can apply~\cref{thm:distr-testing} with domain size $k$, $b\leftarrow \sqrt{\frac{2\dims}{k(\dims+1)}}$, and $\eps\leftarrow \frac{\eps}{\sqrt{k(\dims+1)}}$. The number of samples $\ns$ required is
    \[
    \ns=\bigO{\sqrt{\frac{2\dims}{k(\dims+1)}}\cdot\frac{k(\dims+1)}{\eps^2}}=\bigO{\frac{\sqrt{k\dims(\dims+1)}}{\eps}}.
    \]
\end{proof}

The upper bound part of~\cref{thm:main} is an immediate corollary of the above theorem.
\begin{corollary}
    If the size of the proper 2-design in~\cref{alg:no-shared} is $k=O(\dims^2)$, then $\ns=O(\dims^2/\eps^2)$ copies are sufficient for~\cref{alg:no-shared}. Specifically, when $\dims$ is a prime power, such 2-design exists due to maximal MUB which satisfies $k=\dims(\dims+1)$.
\end{corollary}
This result suggests that the optimal copy complexity of $O(\dims^2/\eps^2)$ can be generalized to dimensions $\dims$ other than prime powers. For example, SIC-POVM~\cite{zauner1999grundzuge,renes2004symmetric} is a minimal 2-design with $k=\dims^2$ and is known to exist for $\dims=2$ to $28$ and as high as $\dims=1299$~\cite{debrota2020informationally}. It has been conjectured in~\cite{zauner1999grundzuge} that SIC-POVMs exist for all $\dims$. If the conjecture is proved, then~\cref{alg:no-shared} naturally generalizes to all $\dims$.

\section*{Acknowledgement} The authors thank Mark Wilde for helpful discussions.

\bibliography{refs}
\bibliographystyle{alpha}

\appendix
\section{Missing proofs in the lower bound}
\label{app:lower-bound-proof}
\subsection{Proof of~\cref{lem:le-cam-approx-eps-perturbation}}
\label{app:lem:le-cam-approx-eps-perturbation}
\begin{proof}
    Recall that $Y=0$ and $\rho=\qmm$ with probability 1/2 and $Y=1$ and $\rho\sim \mathcal{D}$ with probability $1/2$. In the former case when the state is $\qmm$ and $Y=0$, then the tester outputs the correct answer with probability at least $2/3$, 
    \[
    \Pr[\hat{Y}=0|Y=0]\ge 2/3.
    \]
    
    When $\rho\sim\mathcal{D}$, note that by the definition of almost-$\eps$ perturbations, the probability that $\|\sigma_z-\qmm\|_1>\eps$ is at least $1/2$. Denote this event as $E$, then $\Pr[E|Y=1]\ge 1/2$ . We can lower bound the success probability as
    \[
    \Pr[\hat{Y}=1|Y=1]\ge \Pr[Y=1|E, Y=1)]\Pr[E|Y=1]\ge \frac{2}{3}\cdot \frac{1}{2}=\frac{1}{3}.
    \]
    Combining the two parts,
    \[
    \Pr[Y=\hat{Y}]=\frac{1}{2}\Pr[\hat{Y}=0|Y=0]+\frac{1}{2}\Pr[\hat{Y}=1|Y=1]\ge \frac{1}{2}\Paren{\frac{2}{3}+\frac{1}{3}}=\frac{1}{2}.
    \]
    By standard argument on the distinguishability of two distributions, $\totalvardist{\expectDistrOf{z}{\bfP_{\sigma_z}}}{\bfP_{\qmm}}\ge 1/2$. Finally, the inequality follows by Pinsker's inequality and the relation between KL and chi-square divergences.

    \begin{equation*}
        \totalvardist{\expectDistrOf{\sigma}{\bfP_{\sigma}}}{\bfP_{\qmm}}\le \sqrt{\frac{1}{2}\kldiv{\expectDistrOf{\sigma\sim\mathcal{D}}{\bfP_{\sigma}}}{\bfP_{\qmm}}}\le\sqrt{\frac{1}{2}\chisquare{\expectDistrOf{\sigma\sim\mathcal{D}}{\bfP_{\sigma}}}{\bfP_{\qmm}}}.\qedhere
    \end{equation*}
\end{proof}

\subsection{Proof of Lemma~\ref{lem:supop-eigendecomposition}}

Let us recall~\cref{lem:supop-eigendecomposition}.
\eigendecomp*
The proof is broken into two parts. In~\ref{sec:part-one} we state some properties of superoperators, and in~\ref{sec:part-two} we provide a proof of the lemma.

\subsubsection{Important properties of $\Luders_{\POVM}$}
\label{sec:part-one}
 We start with some useful definitions. 

\begin{definition} Let $\mathcal{N}:\C^{\dims\times \dims}\mapsto \C^{\dims\times \dims}$ be a superoperator.
    \begin{enumerate}
        \item $\mathcal{N}$ is called \textit{Hermitian} if $\mathcal{N}=\mathcal{N}^\dagger$.
        \item $\mathcal{N}$ is \textit{Hermiticity preserving} if for all Hermitian $X\in \Herm{\dims}$, $\mathcal{N}(X)$ is also Hermitian.
        \item $\mathcal{N}$ is \textit{trace-preserving} if for all $X\in \C^{\dims\times\dims}$, $\Tr[X]=\Tr[\mathcal{N}(X)]$.
        \item $\mathcal{N}$ is \textit{unital} if $\mathcal{N}(\eye_\dims)=\eye_\dims$.
    \end{enumerate}
    \label{def:supop-properties}
\end{definition}
\label{app:lem:supop-eigendecomposition}

We have the following fact about the L\"uders channel.
\begin{fact}
\label{fact:expect-op-property}
$\Luders_{\POVM}$ is a superoperator over $\C^{\dims\times \dims}$ that satisfies all properties in~\cref{def:supop-properties}.
\end{fact}
\begin{proof}
    The proof follows from Definition~\ref{def:supop-properties}, Definition~\ref{def:expected-density}, and the definition of POVMs. Nevertheless, we provide the proof for completeness.
    \begin{enumerate}
        \item Hermitian: 
        \begin{align*}
            \hdotprod{Y}{\Luders_{\POVM}(X)}&=\Tr\left[Y^\dagger\sum_{x}\frac{\qproj{\psi_x}X\qproj{\psi_x}}{\qdotprod{\psi_x}{\psi_x}}\right]=\sum_{x}\Tr\left[\frac{Y^\dagger\qproj{\psi_x}X\qproj{\psi_x}}{\qdotprod{\psi_x}{\psi_x}}\right]\\
            &=\sum_{x}\Tr\left[\frac{\qproj{\psi_x}Y^\dagger\qproj{\psi_x}X}{\qdotprod{\psi_x}{\psi_x}}\right]=\Tr\left[\sum_{x}\frac{\qproj{\psi_x}Y^\dagger\qproj{\psi_x}}{\qdotprod{\psi_x}{\psi_x}}X\right]\\
            &=\hdotprod{\Luders_{\POVM}(Y)}{X}
        \end{align*}
        \item Hermiticity preserving: let $X$ be Hermitian, then
        \begin{align*}
            \Luders_{\POVM}(X)^\dagger = \sum_{x}\frac{\qproj{\psi_x}X^\dagger\qproj{\psi_x}}{\qdotprod{\psi_x}{\psi_x}}=\sum_{x}\frac{\qproj{\psi_x}X\qproj{\psi_x}}{\qdotprod{\psi_x}{\psi_x}}=\Luders_{\POVM}(X)
        \end{align*}
        \item Trace preserving: 
        \begin{align*}
            \Tr\left[\Luders_{\POVM}(X)\right]&=\Tr\left[\sum_{x}\frac{\qproj{\psi_x}X\qproj{\psi_x}}{\qdotprod{\psi_x}{\psi_x}}\right]\\
            &=\sum_{x}\Tr\left[\frac{\qproj{\psi_x}X\qproj{\psi_x}}{\qdotprod{\psi_x}{\psi_x}}\right]\\
            &=\sum_{x}\matdotprod{\psi_x}{X}{\psi_x}=\sum_{x}\Tr[\qproj{\psi_x}[X]]\\
            &=\Tr\left[\sum_{x}\qproj{\psi_x}X\right]=\Tr[X].
        \end{align*}
        \item Unital: 
        \begin{align*}
            \Luders_{\POVM}(\eye_\dims)&=\sum_{x}\frac{\qproj{\psi_x}\eye_\dims\qproj{\psi_x}}{\qdotprod{\psi_x}{\psi_x}}=\sum_{x}\qproj{\psi_x}=\eye_\dims.\qedhere
        \end{align*}
    \end{enumerate}
\end{proof}

\subsubsection{Proof of the lemma}
\label{sec:part-two}
\paragraph{Hermitian eigenvectors.}

  By linearity, $\avgLuders$ satisfies all properties in Fact~\ref{fact:expect-op-property}. Since $\avgLuders$ is Hermiticity preserving, $\avgLuders$ is also a linear superoperator over the subspace of all Hermitian matrices $\Herm{\dims}$.

    Since $\avgLuders$ is a Hermitian operator on $\Herm{\dims}$, the eigenvectors of $\avgLuders$ form an orthonormal basis $\{V_i\}_{i=1}^{\dims^2}$ of $\Herm{\dims}$. 
    Note that $\eye_\dims$ is an eigenvector of $\bar{\mathcal{H}}$ with eigenvalue 1 since
    \[
    \bar{\mathcal{H}}\left(\eye_d\right)=\frac{1}{\ns}\Paren{\sum_{i=1}^{\ns}\Luders_i(\eye_d)}=\frac{1}{\ns}\Paren{\sum_{i=1}^{\ns}\eye_d=\eye_d}.
    \]
    We then set $V_{\dims^2}=\eye_\dims/\sqrt{\dims}$. Thus, all other eigenvectors $V_1, \ldots, V_{\dims^2-1}$ must lie in the space orthogonal to $\Span\{\eye_d\}$, which is exactly the space of trace-0 Hermitian matrices since
    \[
    \hdotprod{A}{\eye_\dims}=0\iff \Tr[A^\dagger\eye_d]=\Tr[A^\dagger]=0=\Tr[A].
    \]

\paragraph{Non-negative eigenvalues.}

    To show that all eigenvalues are non-negative, we just need to show that $\avgLuders$ is positive semi-definite, i.e. for all matrix $X\in C^{\dims\times \dims}$,
    \[
    \hdotprod{X}{\avgLuders(X)}\ge 0. 
    \]
    Due to linearity, we just need to prove that each $\mathcal{H}_i$ as defined in~\ref{equ:rank-1-luders-channel} is p.s.d.,
    \begin{align}
        \hdotprod{X}{{\Luders_i}(X)}&=\Tr\left[X^\dagger\sum_{x=1}^{k}\frac{\qproj{\psi_x^i}X\qproj{\psi_x^i}}{\qdotprod{\psi_x^i}{\psi_x^i}}\right]\nonumber\\
        &=\sum_{x=1}^{k}\frac{\Tr[X^\dagger\qproj{\psi_x^i}X\qproj{\psi_x^i}]}{\qdotprod{\psi_x^i}{\psi_x^i}}\label{equ:psd-step}\\
        &=\sum_{x=1}^{k}\frac{\matdotprod{\psi_x^i}{X^\dagger}{\psi_x^i}\matdotprod{\psi_x^i}{X}{\psi_x^i}}{\qdotprod{\psi_x^i}{\psi_x^i}}\nonumber\\
        &=\sum_{x=1}^{k}\frac{|\matdotprod{\psi_x^i}{X}{\psi_x^i}|^2}{\qdotprod{\psi_x^i}{\psi_x^i}}\ge 0.\nonumber
    \end{align}
    The last line is due to 
    \[
    \overline{\matdotprod{\psi_x^i}{X}{\psi_x^i}}=\matdotprod{\psi_x^i}{X}{\psi_x^i}^\dagger =\matdotprod{\psi_x^i}{X^\dagger}{\psi_x^i}.
    \]

\paragraph{Upper bound on eigenvalues.}

    Finally, we show that all eigenvalues are at most 1. This is equivalent to $\opnorm{\avgLuders}\le 1$. By the convexity of norms, it suffices to prove that $\opnorm{\Luders_i}\le 1$. Starting from~\eqref{equ:psd-step},
    \begin{align*}
        \hdotprod{X}{{\Luders_i}(X)}
        &=\sum_{x=1}^{k}\frac{\Tr[X^\dagger\qproj{\psi_x^i}X\qproj{\psi_x^i}]}{\qdotprod{\psi_x^i}{\psi_x^i}}\\
        &\le \sum_{x=1}^{k}\frac{\sqrt{\Tr[X^\dagger\qbit{\psi_x^i}\qdotprod{\psi_x^i}{\psi_x^i}\qadjoint{\psi_x^i}X]\Tr[\qproj{\psi_x^i}X^\dagger X\qproj{\psi_x^i}]}}{\qdotprod{\psi_x^i}{\psi_x^i}}&\text{Cauchy-Schwarz}\\
        &=\sum_{x=1}^k\sqrt{\Tr[X^\dagger X\qproj{\psi_i^x}]\Tr[X X^\dagger\qproj{\psi_i^x}]}&\text{Cyclicity of trace}\\
        &\le \sum_{x=1}^k\Tr\left[\frac{X^\dagger X+XX^\dagger}{2}\qproj{\psi_i^x}\right]&\text{AM-GM}\\
        &=\Tr\left[\frac{X^\dagger X+XX^\dagger}{2}\sum_{x=1}^k\qproj{\psi_i^x}\right]\\
        &=\Tr\left[\frac{X^\dagger X+XX^\dagger}{2}\right]& \text{POVM}\\
        &=\Tr[X^\dagger X]\\
        &=\hdotprod{X}{X}.
    \end{align*}
    
\paragraph{Trace.}

Again due to linearity, we only need to prove that $\Tr[{\Luders_{l}}]=\dims$ for each $l=1, \ldots, n$.
    \begin{align*}
        \Tr[{\Luders_{l}}]&=\sum_{i, j=1}^{\dims}\Tr\left[\sum_{x=1}^{k}\frac{\qbit{j}\qdotprod{i}{\psi_x^l}\qdotprod{\psi_x^l}{i}\qdotprod{j}{\psi_x^l}\qadjoint{\psi_x^l}}{\qdotprod{\psi_x^l}{\psi_x^l}}\right]\\
        &=\sum_{i, j=1}^{\dims}\sum_{x=1}^{k}\frac{\qdotprod{\psi_x^l}{j}\qdotprod{i}{\psi_x^l}\qdotprod{\psi_x^l}{i}\qdotprod{j}{\psi_x^l}}{\qdotprod{\psi_x^l}{\psi_x^l}}\\
        &=\sum_{x=1}^k\frac{1}{\qdotprod{\psi_x^l}{\psi_x^l}}\sum_{i=1}^{\dims}|\qdotprod{i}{\psi_x^l}|^2\sum_{j=1}^{\dims}|\qdotprod{j}{\psi_x^l}|^2\\
        &=\sum_{x=1}^k\frac{\qdotprod{\psi_x^l}{\psi_x^l}^2}{\qdotprod{\psi_x^l}{\psi_x^l}}\\
        &=\sum_{x=1}^k\qdotprod{\psi_x^l}{\psi_x^l}=\dims.
    \end{align*}
    The final equality is due to $\sum_{x}\qproj{\psi_x^l}=\eye_d$ and thus $\Tr[\sum_{x}\qproj{\psi_x^l}]=\sum_{x=1}^k\qdotprod{\psi_x^l}{\psi_x^l}=\Tr[\eye_d]=\dims$

\subsection{Proof of~\cref{thm:rand-mat-opnorm-concentration}}
\label{app:prop:perturbation-trace-distance}

We first restate the theorem.
\randmatopnorm*
\begin{proof}
    We first prove that for any fixed unit vector $x\in \C^{\dims}$, the norm of $Wx$ is at most $O(\sqrt{\dims})$ with high probability. Then we use an $\epsilon$-net argument to show that the probability is also high for \textit{all} unit vectors. We start with the following lemma.
    \begin{lemma}
    \label{lem:fixed-vec-opnorm-concentration}
        Let $\{\ptb_i\}_{i=1}^{\dims^2}, \{V_i\}_{i=1}^{\dims^2}$ and $W$ be defined in~\cref{thm:rand-mat-opnorm-concentration}. Then there exists a universal constant $c'$ for any fixed unit vector $x$ and all $s>0$,
        \[
        \probaOf{\normtwo{Wx}\ge (1+s)\sqrt{\dims}}\le 2\exp\{-c's^2\dims\}.
        \]
    \end{lemma}
    \begin{proof}
        Let $\ptb = (\ptb_1, \ldots, \ptb_{\dims^2})\in \R^{\dims^2}$, and $\Pi_\ell\in \R^{\dims^2\times\dims^2}$ be a diagonal matrix with 1 in the first $\ell$ diagonal entries and 0 everywhere else. Then 
        \[
        Wx=\sum_{i=1}^{\ell}\ptb_iV_ix = V_x\Pi_\ell \ptb,
        \]
        where 
        \begin{equation*}
            V_x\eqdef[V_1x, \ldots, V_{\dims^2}x]\in\C^{\dims\times \dims^2}
        \end{equation*}
        which is an isometry, i.e. $V_xV_x^\dagger = \eye_\dims$, as stated in~\cref{claim:isometry} which will be proved at the end of this section. Therefore,
        \[
        \opnorm{V_x}=1, \quad \hsnorm{V_x}^2=\Tr[V_xV_x^\dagger]=\dims.
        \]From this, we can apply concentration for linear transforms of independent sub-Gaussian random variables.
        \begin{theorem}[{~\cite[Theorem 6.3.2]{vershynin2018high}}]
            Let $B\in \C^{m\times n}$ be a fixed $m\times n$ matrix and let $X=(X_1, \ldots, X_n)\in\R^n$ be a random vector with independent, mean zero, unit variance, and sub-Gaussian coordinates with Orlicz-2 norm $\|X_i\|_{\psi_2}\le K$. Then there exists a universal constant $C=\frac{3}{8}$ such that for all $t>0$,
            \[
            \probaOf{|\normtwo{BX}-\hsnorm{B}|>t}\le 2\exp\left\{-\frac{Ct^2}{K^4\opnorm{B}^2}\right\}.
            \]
        \end{theorem}
        \begin{remark}
            The original~\cite[Theorem 6.3.2]{vershynin2018high} was stated for real matrix $B$. However, it is straightforward to extend the argument to complex $B$ by considering $\tilde{B}=\begin{bmatrix}
                \Real(B)\\
                \Img(B)
            \end{bmatrix}$. Then $\opnorm{\tilde{B}}=\opnorm{B}, \hsnorm{\tilde{B}}=\hsnorm{B}$, and $\normtwo{BX}=\normtwo{\tilde{B}X}$. 
        \end{remark}
        Setting $B=V_x\Pi_\ell$, we observe that
        \[
        \opnorm{B}\le \opnorm{V_x}\opnorm{\Pi_\ell}=1, \quad  \hsnorm{B}\le \hsnorm{V_x}=\sqrt{\dims}.
        \]
        Thus, plugging $t=s\sqrt{\dims}$, and noting that $\|\ptb_i\|_{\psi_2}=1/\sqrt{\ln 2}=K$, we have
        \[
        \probaOf{\normtwo{Wx}> (1+s)\sqrt{\dims}}\le \probaOf{\normtwo{B\ptb}>s\sqrt{\dims}+\hsnorm{B}}\le 2\exp\left\{-C\dims(\ln2)^2s^2\right\}.
        \]   
        Setting $c'=C(\ln2)^2=\frac{3(\ln2)^2}{8}$ completes the proof.
    \end{proof}
    We can then proceed to use the $\epsilon$-net argument, which follows closely to \cite[Section 2.3]{tao2023topics}.

    \begin{lemma}[{\cite[Lemma 2.3.2]{tao2023topics}}]
        Let $\Sigma$ be a maximal $1/2$-net of the unitary sphere, i.e., a maximal set of points that are separated from each other by at least $1/2$. Then for any matrix $M\in\C^{\dims\times\dims}$ and $\lambda>0$,
        \[
         \probaOf{\opnorm{M}>\lambda}\le \sum_{y\in\Sigma}\probaOf{\normtwo{My}>\lambda/2}.
        \]
    \end{lemma}
    By standard volume packing argument, the size of $\Sigma$ is at most $\exp(O(\dims))$, 
    \begin{lemma}[{\cite[Lemma 2.3.4]{tao2023topics}}]
    \label{lem:packing-argument}
        Let $\epsilon\in(0, 1)$ and let $\Sigma$ be an $\epsilon$-net of the unit sphere. Then $|\Sigma|\le (C'/\epsilon)^\dims$ where $C'=3$.
    \end{lemma}
    Thus with $c'$ defined in~\cref{lem:fixed-vec-opnorm-concentration} and $C'$ defined in~\cref{lem:packing-argument} we conclude that
    \[
    \probaOf{\opnorm{W}>2(1+s)\sqrt{\dims}}\le 2(2C')^\dims \exp\{-c's^2\dims\}=2\exp\left\{-(c's^2-\ln(2C'))\dims\right\}.
    \]
    Thus choosing $s$ sufficiently large, we can guarantee that the tail probability decays exponentially in $\dims$. Specifically, let $\alpha>0$ and $s^2=\frac{\alpha+\ln(2C')}{c'}$, then we have
    \[
    \probaOf{\opnorm{W}>2(1+s)\sqrt{\dims}}\le 2e^{-\alpha\dims}.
    \]
    Setting $\cop_\alpha = 2(1+s)=2\Paren{1+\sqrt{\frac{\alpha+\ln(2C')}{c'}}}$ proves the theorem. In particular, $\kappa_1\le 10$ when substituting the values of $c'$ and $C'$.
\end{proof}

We end this section with the proof of the isometry claim.
\begin{claim}
 \label{claim:isometry}
    Let $V_1, \ldots, V_{\dims^2}$ be an orthonormal basis of $\C^{\dims\times \dims}$ and $x\in\C^\dims$ be a unit vector. Then $V_x\eqdef[V_1x, \ldots, V_{\dims^2}x]\in\C^{\dims\times \dims^2}$ is an isometry: $V_xV_x^\dagger = \eye_\dims$.
\end{claim}
\begin{proof}
 Let $V_x\supparen{k}$ be the $k$th row of $V_x$ written as row vector. It suffices to prove that
 \[
 V_x\supparen{k}(V_x\supparen{l})^{\dagger}=\delta_{kl}
 \]
 Let $V_i^{(k)}$ be the $k$th row of $V_i$, written as a row vector. Then the $k$th element of $V_ix$ is
        \begin{equation*}
            \quad v_i\supparen{k}\eqdef V_i^{(k)}x.
        \end{equation*}
     Since $V_1, \ldots, V_{\dims^2}$ are orthonormal, we know that
        \[
        V\eqdef[\VecOp(V_1), \ldots, \VecOp(V_{\dims^2})]
        \]
        is a unitary matrix in $\C^{\dims^2\times \dims^2}$. Let $V^{j}$ be the $j$th row of $V$, then because $V$ is unitary, the vector dot product $\hdotprod{V^{j}}{ V^{i}}=\delta_{ij}$. Let 
        \[
        V^{(k)}=[(V^{k})^{\dagger}, (V^{k+\dims})^{\dagger}, \ldots(V^{k+\dims(j-1)})^{\dagger}, \ldots, (V^{k+\dims(d-1)})^{\dagger}]^\dagger
        \]
        which picks out the $k$th row of all $V_1, \ldots, V_{\dims^2}$. Then, we have
        \[
        V^{(k)}=[(V_1^{(k)})^\top, \ldots, (V_{\dims^2}^{(k)})^\top]. 
        \]
        Thus,
        \[
        \sum_{i=1}^{\dims^2}(V_i^{(k)})^{\dagger}V_i^{(k)}=\overline{V^{(k)}(V^{(k)})^{\dagger}}=\eye_\dims,
        \]
        and for $k\ne l$,
        \[
        \sum_{i=1}^{\dims^2}(V_i^{(k)})^{\dagger}V_i^{(l)}=\overline{V^{(k)}(V^{(l)})^{\dagger}}=0.
        \]
        Therefore, 
        \[
         V_x\supparen{k}(V_x\supparen{l})^{\dagger}=\sum_{i=1}^{\dims^2}v_i\supparen{k}(v_i\supparen{l})^\dagger=\sum_{i=1}^{\dims^2}  x^{\dagger}(V_i^{(l)})^{\dagger}V_i^{(k)}x=x^\dagger\delta_{kl}\eye_\dims x=\delta_{kl},
        \]
        exactly as desired, completing the proof.
\end{proof}

\subsection{Proof of~\cref{thm:chi-square-upper-bound}
}
\label{app:thm:chi-square-upper-bound}
We first recall the theorem. 
\chisqub*

\begin{proof}

    First, we claim that due to the exponentially small probability of the bad event $\Delta_z+\qmm\notin \mathcal{P}_\eps$ as stated in~\cref{prop:perturbation-trace-distance}, we can consider $\Delta_z$ instead of the normalized perturbation $\barDelta_z$. The claim is proved at the end of this section.
    \begin{claim} 
    \label{claim:projected-chi-square-ub}
    Let $\barDelta_z$ and $\Delta_z$ be defined in~\cref{def:perturbation}, then
        $$\expectDistrOf{z, z'}{\exp\left\{\ns\dims \hdotprod{\barDelta_{z'}}{\bar{\mathcal{H}}(\barDelta_z)}\right\}}\le \expectDistrOf{z, z'}{\exp\left\{\ns\dims \hdotprod{\Delta_{z'}}{\bar{\mathcal{H}}(\Delta_z)}\right\}}+\frac{4}{e^\dims}.$$       
    \end{claim}

    We then apply a standard result on the moment generating function of Radamacher chaos.
    \begin{lemma}[{\cite[Claim IV.17]{AcharyaCT19}}]
    \label{lem:radamacher-mgf}
        Let $\theta, \theta'$ be two independent random vectors distributed uniformly over $\{-1, 1\}^\ell$. Then for any positive semi-definite real matrix $H$, 
        \[
        \log\expectDistrOf{\theta,\theta'}{\exp\{\lambda \theta^\top H\theta'\}}\le\frac{\lambda^2}{2}\frac{\|H\|_{HS}^2}{1-4\lambda^2\opnorm{H}^2}, \quad\text{for }0\le \lambda<\frac{1}{\opnorm{H}}.
        \]
    \end{lemma}
    We now evaluate the inner product. Recall the Choi representation of $\avgLuders$ is $\avgChoi=\frac{1}{\ns}\sum_{i=1}^{\ns}\Choi_i$.  Note that $\Choi_i$ and $\avgChoi$ are p.s.d. Hermitian matrices, and the eigenvalues exactly match those of $\Luders_i$ and $\avgLuders$ due to the homomorphism between $\C^{\dims^2}$ and $\C^{\dims\times\dims}$.
    
    Setting $V=[\VecOp(V_1), \ldots, \VecOp(V_\ell)]\in \C^{\dims^2\times \ell}$, we have $\VecOp(\Delta_z)=\frac{\cd\eps}{\sqrt{\dims\ell}}Vz$. Due to the homomorphism,
    \begin{align*}
         \hdotprod{\Delta_{\ptb}}{\avgLuders(\Delta_{\ptb'})}&=\VecOp(\Delta_{\ptb'})^\dagger \avgChoi\VecOp(\Delta_{\ptb'})\\
         &=\frac{\cd^2\eps^2}{\dims\ell}\ptb^\dagger V^\dagger\avgChoi V\ptb'.
    \end{align*}
    We now show that  $H\eqdef V^\dagger\avgChoi V$ is a real matrix when each $V_i$ is a Hermitian matrix. First note that $\avgChoi V=[\VecOp(\avgLuders(V_1)), \ldots, \VecOp(\avgLuders(V_\ell))]$. Therefore the $i,j$ the element in $H$ is
    \begin{equation}
    \label{equ:H-element}
        H_{ij}=\VecOp(V_i)^\dagger\VecOp(\avgLuders(V_j))=\hdotprod{V_i}{\avgLuders(V_j)}\in \R.
    \end{equation}

    We use the fact that $\avgLuders$ is Hermiticity preserving and thus $\avgLuders(V_j)$ is Hermitian. Since $\Herm{\dims}$ is a real Hilbert space, the inner product is a real number.

    We then set $\lambda=\frac{\cd^2\ns\eps^2}{\dims\ell}$ and $H=V^\dagger \avgChoi V$ in~\cref{lem:radamacher-mgf}. Then $\opnorm{H}\le \opnorm{\avgChoi}=\opnorm{\avgLuders}\le 1$ due to~\cref{lem:supop-eigendecomposition}. Thus for $\ns<\frac{\ell}{3\cd^2\eps^2}$, we have
    
    $$\lambda\opnorm{H}\le \lambda<\frac{1}{3}\implies \frac{\lambda^2}{2(1-4\lambda^2\opnorm{H}^2)}\le\frac{9\lambda^2}{10}<\lambda^2. $$
    Hence, applying Lemma~\ref{lem:radamacher-mgf} 
    \begin{align*}
        \expectDistrOf{z, z'}{\exp\left\{\frac{\ns\eps^2}{\ell}z^\top Hz'\right\}}\le\exp\{\lambda^2\hsnorm{H}^2\}= \exp\left\{\frac{\cd^4\ns^2\eps^4}{\dims^2\ell^2}\hsnorm{H}^2\right\}.
    \end{align*}
    Combining with~\cref{claim:projected-chi-square-ub} proves~\cref{thm:chi-square-upper-bound}.     
\end{proof}

\begin{proof}[Proof of~\cref{claim:projected-chi-square-ub}]    
        Note that $\barDelta_z=a_z\Delta_z$, where
    \[
    a_z\eqdef\min\left\{1, \frac{1}{\dims\opnorm{\Delta_z}}\right\}\in[0, 1].
    \]
    Therefore,
    \[
    \hdotprod{\barDelta_{z'}}{\bar{\mathcal{H}}(\barDelta_z)}=a_za_{z'}\hdotprod{\Delta_{z'}}{\bar{\mathcal{H}}(\Delta_z)}.
    \]
        As a short hand let $f(z, z')=\ns\dims\hdotprod{\Delta_{z'}}{\bar{\mathcal{H}}(\Delta_z)}$. Denote event $E$ as $f(z,z')<0$ and $a_za_{z'}< 1$. When this event occors, $\exp\{a_za_{z'}f(z,z')\}\le 1$. Using~\cref{prop:perturbation-trace-distance}, let $\delta= 2\exp(-\dims)$,
        \[
        \probaOf{a_{z}<1}\le \delta.
        \]
        Thus, by the union bound,
        \[
        \probaOf{E}= \probaOf{a_{z}a_{z'}<1}=\probaOf{a_{z}<1 \text{ or } a_{z'}<1}\le 2\delta.
        \]
        Note that $E^c$ denotes the event that $f(z,z')\ge 0$ or $a_za_z'=1$. When this occurs, $a_za_{z'}f(z,z')\le f(z,z')$. Thus,
        \begin{align*}
            &\quad \expectDistrOf{z, z'}{\exp\left\{a_za_{z'}f(z,z')\right\}}\\
            &=\expectCondDistrOf{z,z'}{\exp\left\{ a_za_{z'}f(z,z')\right\}}{E^c}\probaOf{E^c} +\expectCondDistrOf{z,z'}{\exp\left\{ a_za_{z'}f(z,z')\right\}}{E}\probaOf{E}\\
            &\le \expectCondDistrOf{z,z'}{\exp\left\{f(z,z')\right\}}{E^c}\probaOf{E^c}+2\delta\\
            &\le  \expectDistrOf{z,z'}{\exp\left\{f(z,z')\right\}}+2\delta,
        \end{align*}
        as desired. The second-to-last inequality uses $a_zf_z'f(z, z')\le 0$ when event $E$ happens, and the final inequality uses $\exp\{f(z,z')\}>0$ and therefore
        \begin{align*}
            \expectDistrOf{z,z'}{\exp\left\{f(z,z')\right\}}&=\expectCondDistrOf{z,z'}{\exp\left\{f(z,z')\right\}}{E^c}\probaOf{E^c}+\expectCondDistrOf{z,z'}{\exp\left\{f(z,z')\right\}}{E}\probaOf{E}\\
            &\ge \expectCondDistrOf{z,z'}{\exp\left\{f(z,z')\right\}}{E^c}\probaOf{E^c}.
        \end{align*}
        Plugging in the definition of $a_z$ and $f(z,z')$ completes the proof.
    \end{proof}

\subsection{Proof of~\cref{obs:hs-norm-luders-projection}
}
\label{app:obs:hs-norm-luders-projection}
Recall the statement of the observation.
\hsnnorm*

\begin{proof}
 Since $V_1, \ldots, V_{\dims^2}$ is an orthonormal basis, we have $V^\dagger V=\eye_{\ell}$, and thus $V$ is an isometry and $\opnorm{V}=\opnorm{V^\dagger}= 1$. Using $\hsnorm{AB}\le \opnorm{A}\hsnorm{B}$, we obtain
    \[
    \hsnorm{V^\dagger\avgChoi V}\le \hsnorm{\avgChoi}=\hsnorm{\avgLuders}.
    \]
    When $\hbasis=\hbasis_{\avgLuders}$, we note that $V^\dagger\avgChoi V=D_\ell\eqdef\diag\{\lambda_1, \ldots, \lambda_\ell\}$, and  $\hsnorm{D_\ell}^2=\sum_{i=1}^\ell \lambda_i^2$. Indeed, as derived in~\eqref{equ:H-element},
    \begin{align*}
        H_{ij}=\hdotprod{V_i}{\avgLuders(V_j)}=\lambda_j\hdotprod{V_i}{V_j}=\lambda_j\delta_{ij}.
    \end{align*}
    Therefore, 
    \[
    \hsnorm{H}^2\le \begin{cases}
        \hsnorm{D_\ell}^2=\sum_{i=1}^{\ell}\lambda_i^2, &\hbasis=\hbasis_{\avgLuders},\\
        \hsnorm{\avgLuders}^2,&\text{otherwise}.
    \end{cases}\qedhere
    \]
    
\end{proof}

\end{document}